\def\IsExtended{1}
\if 1\IsExtended
\documentclass[envcountsame, a4paper]{article}
\else
\documentclass[orivec,envcountsame, a4paper]{llncs}
\usepackage[paperheight=235mm, paperwidth=155mm,textwidth=12.2cm,textheight=19.3cm]{geometry}
\fi

\newcommand{\ourtitle}{A Sorted Datalog Hammer for Supervisor Verification Conditions Modulo Simple Linear Arithmetic}

\usepackage{amsmath}

\usepackage{amsfonts}
\usepackage{amssymb}
\usepackage{amsthm}

\usepackage[utf8]{inputenc}
\usepackage{newtxtext,newtxmath} 
\DeclareMathAlphabet{\mathcal}{OMS}{cmsy}{m}{n} 
\usepackage{paralist} 
\usepackage{bbding}
\usepackage{graphicx}

 \usepackage{comment}

\usepackage{url}
\usepackage[bookmarks=false,pdfpagelabels,bookmarksnumbered,bookmarksopen,bookmarksopenlevel=1,colorlinks=true,
citecolor=black, filecolor=black, linkcolor=black, menucolor=black,
urlcolor=black, pdftitle={\ourtitle}, pdfauthor={Martin Bromberger, Irina Dragoste, Rasha Faqeh, Christof Fetzer, Larry Gonzalez, Markus Kroetzsch, Maximilian Marx, Harish K Murali, Christoph Weidenbach},
pdfkeywords={Datalog, arithmetic constraints, Bernays-Schoenfinkel first-order logic, SUPERLOG, SMT}]{hyperref}
\usepackage{listings} 
\lstdefinestyle{embedded}{
	numbers=none,
	frame=none,
	xleftmargin=0cm,
	backgroundcolor=\color{Lavender},
	framesep=1pt,
	aboveskip=3pt,
	belowskip=3pt,
    captionpos=b,                    
}
\lstdefinestyle{small}{
	basicstyle=\linespread{0.9}\footnotesize,
}

\mathchardef\hyphenmathcode=\mathcode`\-

\let\origlstlisting=\lstlisting
\let\endoriglstlisting=\endlstlisting
\renewenvironment{lstlisting}
    {\mathcode`\-=\hyphenmathcode
     \everymath{}\mathsurround=0pt\origlstlisting}
    {\endoriglstlisting}

\usepackage{xcolor} 

\usepackage{wrapfig} 

\usepackage[english]{babel}

\theoremstyle{definition}

\if 1\IsExtended

\newenvironment{extendedonly}{}{}
\excludecomment{notinextended}

\theoremstyle{definition}
\newtheorem{theorem}{Theorem}
\newtheorem{lemma}[theorem]{Lemma}

\newtheorem{example}[theorem]{Example}
\newtheorem{definition}[theorem]{Definition}

\else

\excludecomment{prooftoggle}
\excludecomment{extendedonly}
\newenvironment{notinextended}{}{}

\usepackage[subtle]{savetrees}

\fi

\newcommand{\RealS}{\mathcal{R}}
\newcommand{\IntS}{\mathcal{Z}}
\newcommand{\FolS}{\mathcal{F}}
\newcommand{\Fol}{\mathbb{F}}

\newcommand{\dval}{\ensuremath{\operatorname{dvals}}}
\newcommand{\aval}{\ensuremath{\operatorname{avals}}}
\newcommand{\dfacts}{\ensuremath{\operatorname{dfacts}}}

\newcommand{\depend}{\ensuremath{\operatorname{depend}}}
\newcommand{\connectedArgs}{\ensuremath{\operatorname{conArgs}}}
\newcommand{\connectedIneqs}{\ensuremath{\operatorname{conIneqs}}}

\newcommand{\sort}{\ensuremath{\operatorname{sort}}}

\newcommand{\tpfunction}{\ensuremath{\operatorname{tp-function}}}
\newcommand{\tpfunctions}{\ensuremath{\operatorname{tp-functions}}}
\newcommand{\wti}{\ensuremath{\operatorname{wti}}}
\newcommand{\tg}{\ensuremath{\operatorname{wtis}}}
\newcommand{\epfunction}{\ensuremath{\operatorname{ep-function}}}
\newcommand{\epfunctions}{\ensuremath{\operatorname{ep-functions}}}

\newcommand{\iPart}{\ensuremath{\operatorname{iPart}}}
\newcommand{\testpoints}{\ensuremath{\operatorname{tps}}}
\newcommand{\iBorders}{\ensuremath{\operatorname{iEP}}}

\newcommand{\aGnd}{\ensuremath{\operatorname{agnd}}}

\newcommand{\AG}{\ensuremath{\operatorname{A}}}

\newcommand{\PAG}{\ensuremath{\operatorname{PA}}}

\newcommand{\Real}{\mathbb{R}}

\newcommand{\Int}{\mathbb{Z}}
\newcommand{\Nat}{\mathbb{N}}

\newcommand{\varset}{\mathcal{X}}

\newcommand{\inta}{\mathcal{A}}
\newcommand{\sigval}{\mathcal{A}}
\newcommand{\vars}{\ensuremath{\operatorname{vars}}}
\newcommand{\atoms}{\ensuremath{\operatorname{atoms}}}

\newcommand{\dom}{\operatorname{dom}}
\newcommand{\cdom}{\operatorname{codom}}
\newcommand{\comp}{\operatorname{comp}} 
\newcommand{\mGnd}{\operatorname{gnd}}
\newcommand{\mMGU}{\operatorname{mgu}} 

\newcommand{\dblquotes}{}

\newcommand{\tren}{\operatorname{tren}}
\newcommand{\tfacts}{\operatorname{tfacts}}
\newcommand{\sfacts}{\operatorname{sfacts}}

\newcommand{\LA}{\ensuremath{\operatorname{LA}}}

\newcommand{\SB}{\ensuremath{\operatorname{SLA}}}
\newcommand{\BS}{\ensuremath{\operatorname{BS}}}

\newcommand{\SP}{\ensuremath{\operatorname{P}}}
\newcommand{\LAOP}{\ensuremath{\operatorname{\triangleleft}}}

\newcommand{\HBS}{\ensuremath{\operatorname{HBS}}}

\newcommand{\SpeedTable}{\ensuremath{\operatorname{SpeedTable}}}

\newcommand{\Speed}{\ensuremath{\operatorname{Speed}}}

\newcommand{\IgnDeg}{\ensuremath{\operatorname{IgnDeg}}}

\newcommand{\ResDegArgs}{\ensuremath{\operatorname{ResArgs}}}
\newcommand{\Conjecture}{\ensuremath{\operatorname{Conj}}}
\newcommand{\Rpm}{\operatorname{Rpm}}

\newcommand{\myparagraph}[1]{\smallskip \noindent{\textbf{#1}}}

\newcommand{\figref}[1]{Fig.~\ref{fig:#1}}

\newcommand{\tspace}[1]{ 
\newcount\foo
\foo=#1
\loop
{\,}
\advance \foo -1
\ifnum \foo>0
\repeat
}

\definecolor{commentcolor}{gray}{0.5}

\title{\ourtitle}
\if 1\IsExtended
\author{Martin Bromberger\\ Max Planck Institute for Informatics\\ Saarland Informatics Campus,  Saarbr\"ucken, Germany\\  \and
  Irina Dragoste\\ TU Dresden, Dresden, Germany\\ \and
  Rasha Faqeh\\ TU Dresden, Dresden, Germany\\ \and
  Christof Fetzer\\ TU Dresden, Dresden, Germany\\ \and
  Larry Gonz\'alez\\ TU Dresden, Dresden, Germany\\ \and
  Markus Kr\"otzsch\\ TU Dresden, Dresden, Germany\\ \and
  Maximilian Marx\\ TU Dresden, Dresden, Germany\\ \and
  Harish K Murali\\ Max Planck Institute for Informatics\\ Saarland Informatics Campus, Saarbr\"ucken, Germany\\ IIITDM Kancheepuram, Chennai, India\\ \and
  Christoph Weidenbach \\ Max Planck Institute for Informatics\\ Saarland Informatics Campus, Saarbr\"ucken, Germany\\}
\else
\author{Martin~Bromberger\inst{1,\Envelope}\and Irina~Dragoste\inst{2} \and Rasha~Faqeh\inst{2} \and Christof~Fetzer\inst{2} \and Larry Gonz\'alez\inst{2} \and Markus~Kr\"otzsch\inst{2} \and Maximilian~Marx\inst{2} \and Harish K Murali\inst{1,3} \and Christoph~Weidenbach\inst{1}}

\institute{Max Planck Institute for Informatics, Saarland Informatics Campus, Saarbr\"ucken, Germany\and TU Dresden, Dresden, Germany\and IIITDM Kancheepuram, Chennai, India}
\titlerunning{A Sorted Datalog Hammer}
\fi

\begin{document}

\maketitle

\begin{abstract}
  In a previous paper, we have shown that clause sets belonging to the Horn Bernays-Sch\"onfinkel fragment over simple linear real arithmetic (HBS(SLR))
  can be translated into HBS clause sets over a finite set of first-order constants.
  The translation preserves validity and satisfiability and 
  it is still applicable if we extend our input with positive universally or existentially quantified verification conditions (conjectures).
  We call this translation a Datalog hammer. The combination of its implementation in SPASS-SPL with the Datalog reasoner VLog
  establishes an effective way of deciding verification conditions in the Horn fragment.
  We verify supervisor code for two examples: a lane change assistant in a car and an electronic control unit of a supercharged combustion engine.
  
  In this paper, we improve our Datalog hammer in several ways: 
  we generalize it to mixed real-integer arithmetic and finite first-order sorts;
  we extend the class of acceptable inequalities beyond variable bounds and positively grounded inequalities; and 
  we significantly reduce the size of the hammer output by a soft typing discipline.
  We call the result the sorted Datalog hammer. It not only allows us to handle more complex supervisor code and
  to model already considered supervisor code more concisely, 
  but it also improves our performance on real world benchmark examples. Finally, we replace the before
  file-based interface between SPASS-SPL and VLog by a close coupling resulting in a single executable binary.
\end{abstract}

\section{Introduction} \label{sec:intro}

Modern dynamic dependable systems (e.g., autonomous driving) continuously update software components to fix bugs and to introduce new features.
However, the safety requirement of such systems demands software to be safety certified before it can be used, which is typically a lengthy process
that hinders the dynamic update of software.
We adapt the \emph{continuous certification} approach \cite{FaqehFH0KKSW20} for variants of safety critical software components
using a \emph{supervisor} that guarantees important aspects through \emph{challenging}, see \figref{supervisor-arch}.
Specifically, multiple processing units run in parallel -- \emph{certified} and \emph{updated not-certified} variants that produce output as \emph{suggestions} and \emph{explications}.
The supervisor compares the behavior of variants and analyses their explications. The supervisor itself consists of a rather small set
of rules that can be automatically verified and run by a reasoner such as SPASS-SPL. In this paper we concentrate on the further development
of our verification approach through the sorted Datalog hammer.

\begin{extendedonly}
\begin{figure}[t]
	\begin{center}
		\includegraphics[scale=0.6]{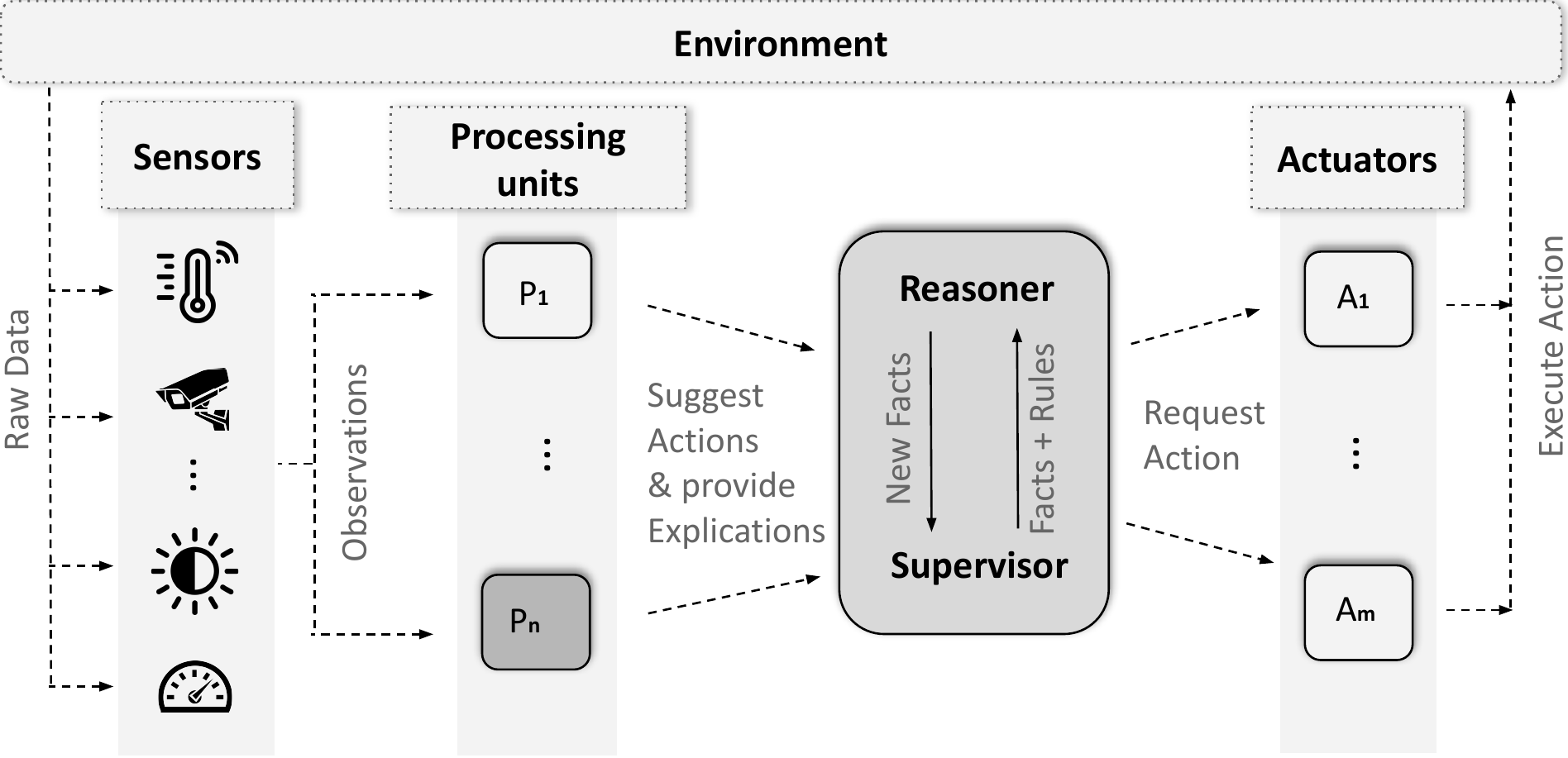}
        \caption{The supervisor architecture.}
        \label{fig:supervisor-arch}
    \end{center}
\end{figure}
\end{extendedonly}
\begin{notinextended}
\begin{figure}[h]
	\begin{center}
	     \vspace{-20pt}
		\includegraphics[scale=0.5]{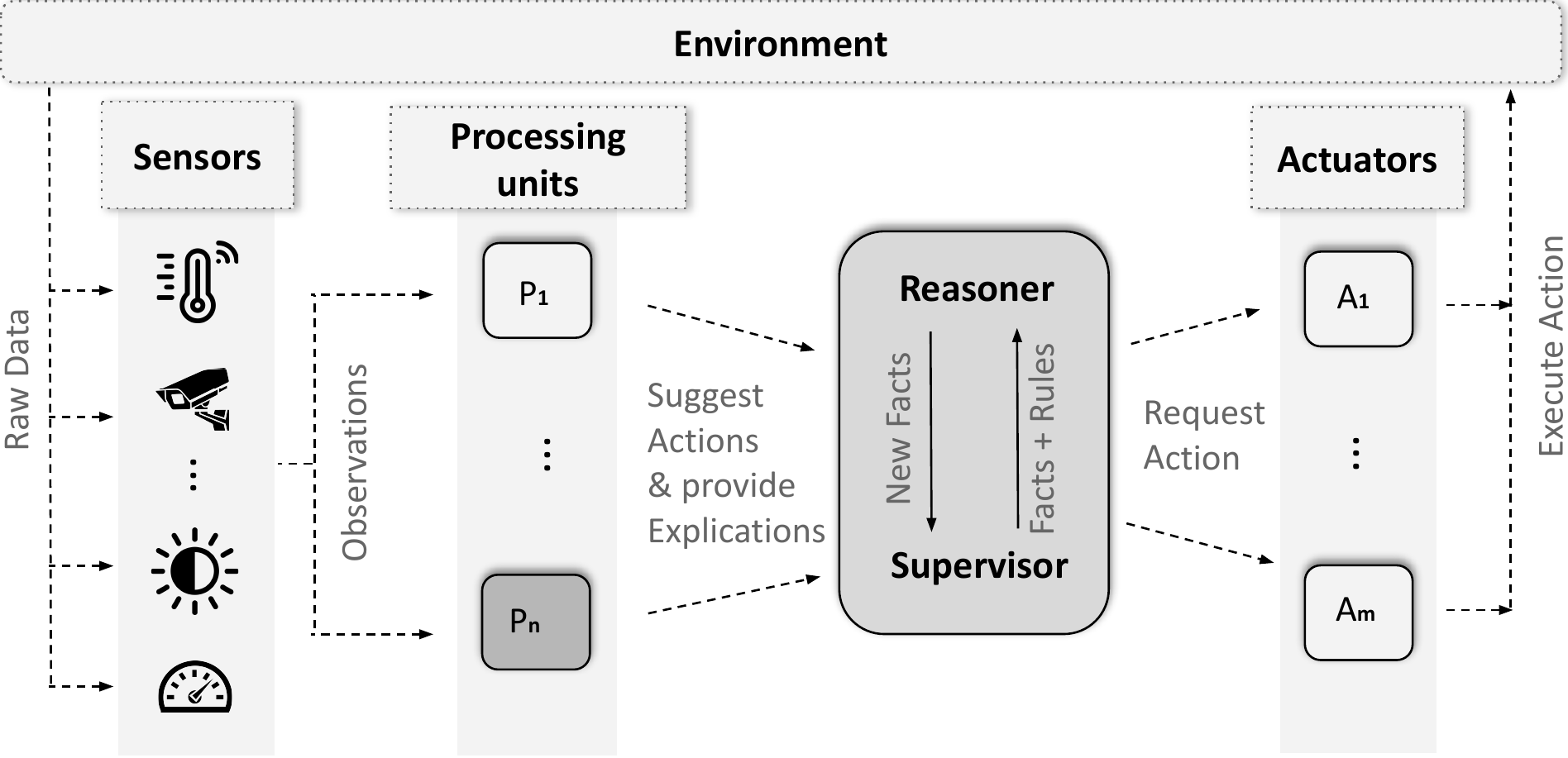}
        \caption{The supervisor architecture.}
        \label{fig:supervisor-arch}
    \end{center}
\end{figure}
\end{notinextended}

While supervisor safety
conditions formalized as existentially quantified properties can often
already be automatically verified,
conjectures about invariants requiring universally quantified properties are a further challenge.
Analogous to the Sledgehammer project~\cite{BoehmeNipkow10} of Isabelle~\cite{NipkowEtAl02} that translates higher-order logic
conjectures to first-order logic (modulo theories) conjectures, our sorted Datalog hammer translates first-order
Horn logic modulo arithmetic conjectures into pure Datalog programs, which is equivalent to the Horn Bernays-Sch\"onfinkel clause fragment, called $\HBS$.

More concretely, the underlying logic for both formalizing supervisor behavior and formulating
conjectures is the hierarchic combination
of the Horn Bernays-Sch\"onfinkel fragment with linear arithmetic, $\HBS(\LA)$,
also called \emph{Superlog} for Supervisor Effective Reasoning Logics~\cite{FaqehFH0KKSW20}.
Satisfiability of $\BS(\LA)$ clause sets is undecidable~\cite{Downey1972,HorbachEtAl17ARXIV}, in general, however, the restriction
to simple linear arithmetic $\BS(\SB)$ yields a decidable fragment~\cite{GeMoura09,HorbachEtAl17CADE}.

Inspired by the test point method for quantifier elimination in arithmetic~\cite{LoosWeispfenning93} we
show that instantiation with a finite number of values is sufficient to decide
whether a universal or existential conjecture is a consequence of a $\BS(\SB)$ clause set.

In this paper, we improve our Datalog hammer~\cite{BrombergerEtAl21FROCOS} for $\HBS(\SB)$ in three directions. 
First, we modify our Datalog hammer so it also accepts other sorts for variables besides reals:
the integers and arbitrarily many finite first-order sorts $\FolS_i$. Each non-arithmetic sort has a predefined finite domain corresponding
to a set of constants $\Fol_i$ for $\FolS_i$ in our signature.
Second,
we modify our Datalog hammer so it also accepts more general inequalities than simple linear arithmetic allows (but only under certain conditions).
In~\cite{BrombergerEtAl21FROCOS}, we have already started in this direction by extending the input logic from pure $\HBS(\SB)$ to pure positively grounded $\HBS(\SB)$.
Here we establish a soft typing discipline by
efficiently approximating potential values occurring at predicate argument positions of all derivable facts.
Third, 
we modify the test-point scheme that is the basis of our Datalog hammer so it can exploit the fact that not all all inequalities are connected to all predicate argument positions.

Our modifications have three major advantages:
first of all, they allow us to express supervisor code for our previous use cases more elegantly and without any additional preprocessing.
Second of all, they allow us to formalize supervisor code that would have been out of scope of the logic before.
Finally, they reduce the number of required test points, which leads 
to smaller transformed formulas that can be solved in much less time.

For our experiments of the test point approach we consider again two case studies. First, verification
conditions for a supervisor taking care of multiple software variants of a lane change assistant. Second, verification conditions for a supervisor 
of a supercharged combustion engine, also called an ECU for Electronical Control Unit. The supervisors in both cases
are formulated by $\BS(\SB)$ Horn clauses.
Via our test point technique they are
translated together with the verification conditions to Datalog~\cite{Alice} ($\HBS$).
The translation is implemented in our Superlog reasoner SPASS-SPL. The resulting Datalog
clause set is eventually explored by the Datalog engine VLog~\cite{Rulewerk2019}.
This
hammer constitutes a decision procedure for both universal and existential conjectures.
The results of our experiments show that we can verify non-trivial existential
and universal conjectures in the range of seconds while state-of-the-art solvers
cannot solve all problems in reasonable time, see Section~\ref{sec:experiments}.

\myparagraph{Related Work:} Reasoning about $\BS(\LA)$ clause sets is supported by SMT (Satisfiability Modulo Theories)~\cite{NieuwenhuisEtAl06,MouraBjorner11}.
In general, SMT comprises the combination of a number of theories beyond $\LA$ such as arrays, lists, strings, or bit vectors.
While SMT is a decision procedure for the $\BS(\LA)$ ground case, universally quantified variables can be considered by instantiation~\cite{ReynoldsEtAl18}.
Reasoning by instantiation does result in a refutationally complete procedure for $\BS(\SB)$, but not in a decision procedure. The Horn
fragment $\HBS(\LA)$ out of $\BS(\LA)$ is receiving additional attention~\cite{GrebenshchikovEtAl12,BjornerEtAl15}, because it is well-suited for 
software analysis and verification. Research in this direction also goes beyond the theory of $\LA$ and considers minimal model semantics in addition, but is restricted
to existential conjectures.
Other research focuses on universal conjectures, but over non-arithmetic theories, e.g., invariant checking for array-based systems~\cite{CimattiGR21}
or considers abstract decidability criteria incomparable with the $\HBS(\LA)$ class~\cite{Ranise12}.
Hierarchic superposition~\cite{BachmairGanzingerEtAl94} and Simple Clause Learning over Theories (SCL(T))~\cite{BrombergerFW21} are both
refutationally complete for $\BS(\LA)$. While SCL(T) can be immediately turned into a decision procedure for even larger fragments than $\BS(\SB)$~\cite{BrombergerFW21}, hierarchic
superposition needs to be refined to become a decision procedure already because of the Bernays-Sch\"onfinkel part~\cite{HillenbrandWeidenbach13}.
Our Datalog hammer translates $\HBS(\SB)$ clause sets with both existential and universal conjectures into $\HBS$ clause sets which
are also subject to first-order theorem proving. Instance generating approaches such as  iProver~\cite{Korovin08} are a decision procedure
for this fragment, whereas superposition-based~\cite{BachmairGanzingerEtAl94} first-order provers
such as E~\cite{SchulzEtAl19}, SPASS~\cite{WeidenbachEtAlSpass2009}, Vampire~\cite{RiazanovVoronkov02},
have additional mechanisms implemented to decide $\HBS$.
In our experiments, Section~\ref{sec:experiments}, we will discuss the differences between all these approaches on a number of benchmark examples in more detail.

The paper is organized as follows: after a section on preliminaries, Section~\ref{sec:prelims}, we present the theory
of our sorted Datalog hammer in Section~\ref{sec:typing}, followed by
experiments on real world supervisor verification conditions, Section~\ref{sec:experiments}. The paper ends with a discussion of
the obtained results and directions for future work, Section~\ref{sec:conclusion}.
\begin{notinextended}
    The artifact (including binaries of our tools and all benchmark problems)
    is available at~\cite{GITEvalSortedHammer}.
    An extended version is available at~\cite{TODO} including proofs and pseudo-code algorithms for the presented results.
\end{notinextended}
\begin{extendedonly}
  This paper is an extended version of~\cite{BrombergerEtAl22Tacas}. 
  The appendix contains proofs and pseudo-code algorithms for the presented results. 
  The artifact (including binaries of our tools and all benchmark problems) is available at~\cite{GITEvalSortedHammer}.
  
\end{extendedonly}


\section{Preliminaries}\label{sec:prelims}

We briefly recall the basic logical formalisms and notations we build upon~\cite{BrombergerEtAl21FROCOS}.
Starting point is a standard many-sorted first-order language for \BS{} with \emph{constants} (denoted $a, b, c$), without non-constant function symbols,
\emph{variables} (denoted $w,x,y,z$), and \emph{predicates} (denoted $P, Q, R$)
of some fixed \emph{arity}.
\emph{Terms} (denoted $t,s$) are variables or constants.
We write $\bar{x}$ for a vector of variables, $\bar{a}$ for a vector of constants, and so on.
An \emph{atom} (denoted $A, B$) is an expression $P(\bar{t})$ for a predicate $P$ of arity $n$ and a term list $\bar{t}$
of length $n$.
A \emph{positive literal} is an atom $A$ and a \emph{negative literal} is a negated atom $\neg A$.
We define $\comp(A)=\neg A$, $\comp(\neg A)=A$, $|A|=A$ and $|\neg A|=A$.
Literals are usually denoted $L, K, H$.

A \emph{clause} is a disjunction of literals, where all variables are assumed to be universally
quantified. $C, D$ denote clauses, and $N$ denotes a clause set.
We write $\atoms(X)$ for the set of atoms in a clause or clause set $X$.
A clause is \emph{Horn} if it contains at most one positive literal, and
a \emph{unit clause} if it has exactly one literal.
A clause $A_1\vee \ldots\vee A_n\vee \neg B_1\vee\ldots\vee \neg B_m$ can be written as an implication
$B_1\wedge \ldots\wedge B_m\to A_1\vee\ldots\vee A_n$, still omitting universal quantifiers.
If $Y$ is a term, formula, or a set thereof, $\vars(Y)$ denotes the set of all variables in $Y$,
and $Y$ is \emph{ground} if $\vars(Y)=\emptyset$.
A \emph{fact} is a ground unit clause with a positive literal.

\myparagraph{Datalog and the Horn Bernays-Sch\"onfinkel Fragment:}
The \emph{Horn case of the Bernays-Schönfinkel fragment} (\HBS) comprises all sets of clauses 
with at most one positive literal.
The more general Bernays-Sch\"onfinkel fragment (\BS{}) in first-order logic allows arbitrary \emph{formulas} over atoms,
i.e., arbitrary Boolean connectives and leading existential quantifiers. 
\BS{} formulas can be polynomially transformed into clause sets with common 
syntactic transformations while preserving satisfiability and all entailments that do not refer to auxiliary
constants and predicates introduced in the transformation~\cite{NonnengartW01}.
\BS{} theories in our sense are also known as \emph{disjunctive Datalog programs} \cite{DisjunctiveDatalog97},
specifically when written as implications.
A \HBS{} clause set is also called a \emph{Datalog program}.
Datalog is sometimes viewed as
a second-order language. We are only interested in query answering, which can equivalently be viewed as first-order entailment
or second-order model checking \cite{Alice}.
Again, it is common to write clauses as implications in this case.

Two types of \emph{conjectures}, i.e., formulas we want
to prove as consequences of a clause set, are of particular interest:
\emph{universal} conjectures $\forall \bar{x}. \phi$ and \emph{existential} conjectures $\exists \bar{x}. \phi$,
where $\phi$ is a $\BS$ formula that only uses variables in $\bar{x}$.
We call such a conjecture positive if the formula only uses conjunctions and disjunctions to connect atoms.
Positive conjectures are the focus of our Datalog hammer and they have the useful property that they can be transformed
to one atom over a fresh predicate symbol by adding some suitable Horn clause definitions to our clause set $N$~\cite{NonnengartW01,BrombergerEtAl21FROCOS}.
This is also the reason why we assume for the rest of the paper that all
relevant universal conjectures have the form $\forall \bar{x}. P(\bar{x})$ and existential conjectures the form  $\exists \bar{x}. P(\bar{x})$.

A \emph{substitution} $\sigma$ is a function from variables to terms with a finite domain 
$\dom(\sigma) = \{ x \mid x\sigma \neq x\}$ and codomain $\cdom(\sigma) = \{x\sigma\mid x\in\dom(\sigma)\}$.
We denote substitutions by $\sigma, \delta, \rho$.
The application of substitutions is often written postfix, as in $x\sigma$, and is homomorphically extended to
terms, atoms, literals, clauses, and quantifier-free formulas.
A substitution $\sigma$ is \emph{ground} if $\cdom(\sigma)$ is ground.
Let $Y$ denote some term, literal, clause, or clause set.
$\sigma$ is a \emph{grounding} for $Y$ if $Y\sigma$ is ground, and $Y\sigma$ is a
\emph{ground instance} of $Y$ in this case.
We denote by $\mGnd(Y)$ the set of all ground instances of $Y$, 
and by $\mGnd_B(Y)$ the set of all ground instances over a given set of constants $B$.
The \emph{most general unifier} $\mMGU(Z_1,Z_2)$ of two terms/atoms/literals $Z_1$ and $Z_2$
is defined as usual, and we assume that it does not introduce fresh variables and is idempotent.

We assume a standard many-sorted first-order logic model theory, and write
$\sigval \models \phi$ if an interpretation $\sigval$ satisfies a first-order formula $\phi$.
A formula $\psi$ is a logical consequence of $\phi$, written $\phi\models\psi$, if
$\sigval\models\psi$ for all $\sigval$ such that $\sigval\models\phi$.
Sets of clauses are semantically treated as conjunctions of clauses with all variables quantified
universally.

\myparagraph{$\BS$ with Linear Arithmetic:}
The extension of $\BS$ with linear arithmetic both over real and integer variables, $\BS(\LA)$,
is the basis for the formalisms studied in this paper.
We extend the standard \emph{many-sorted} first-order logic with finitely many first-order sorts $\FolS_i$ and with two arithmetic sorts 
$\RealS$ for the real numbers and $\IntS$ for the integer numbers. The sort $\IntS$ is a \emph{subsort} of $\RealS$.   
Given a clause set $N$, the interpretations $\inta$ of our sorts are fixed: 
$\RealS^{\inta} = \Real$, 
$\IntS^{\inta} = \Int$, 
and $\FolS_i^{\inta} = \Fol_i$, i.e., a first-order sort interpretation $\Fol_i$ consists of the set of constants in $N$ belonging to that sort, or a single constant out of the signature if no
such constant occurs. Note that this is not a deviation from standard semantics in our context as for the arithmetic part the canonical domain is considered and for the first-order sorts BS has
the finite model property over the occurring constants which is sufficent for refutation-based reasoning.
This way first-order constants are distinct values.

Constant symbols, arithmetic function symbols, variables, and predicates are uniquely declared together with \emph{sort} expressions.
The unique sort of a constant symbol, variable, predicate, or term is denoted by the function $\sort(Y)$  and
we assume all terms, atoms, and formulas to be well-sorted. The sort of predicate $P$'s argument position $i$ is denoted by $\sort(P,i)$.
For arithmetic function symbols we consider the minimal sort
with respect to the subsort relation between $\RealS$ and $\IntS$.
Eventually, we don't consider arithmetic functions here, so the subsort relationship boils down
to substitute an integer sort variable or number for a real sort variable.

We assume \emph{pure} input clause sets, which means the only constants of sort $\RealS$ or $\IntS$ are numbers.
This means the only constants that we do allow are integer numbers $c \in \mathbb{Z}$ and the constants defining our finite first-order sorts $\FolS_i$.
Satisfiability of pure $\BS(\LA)$ clause sets is semi-decidable, e.g., using \emph{hierarchic superposition} \cite{BachmairGanzingerEtAl94}
or \emph{SCL(T)}~\cite{BrombergerFW21}.
Impure $\BS(\LA)$ is no longer compact and satisfiability becomes undecidable, 
but it can be made decidable when restricting to ground clause sets~\cite{BrombergerEtAl2020arxiv}.

All arithmetic predicates and functions are interpreted in the usual way.
An interpretation of $\BS(\LA)$ coincides with $\sigval^{\LA}$ on arithmetic predicates and functions, and freely interprets
free predicates. For pure clause sets this is well-defined~\cite{BachmairGanzingerEtAl94}.
Logical satisfaction and entailment is defined as usual, and uses similar notation as for \BS.

\begin{example}\label{ex_bs_lra_ecu}
The following $\BS(\LA)$ clause from our ECU case study compares the
values of engine speed ($\text{Rpm}$) and pressure ($\text{KPa}$) with entries in an ignition
table ($\text{IgnTable}$) to derive the basis of the current ignition value ($\text{IgnDeg1}$):
\begin{align}\begin{split}
  &x_1 < 0 \;\lor\; x_1\geq 13 \;\lor\; x_2 < 880  \;\lor\; x_2\geq 1100 \;\lor\; \neg \text{KPa}(x_3, x_1) \;\lor{} \\
  &  \neg \text{Rpm}(x_4, x_2) \;\lor\; \neg \text{IgnTable}(0,13,880,1100,z)  \;\lor\; \text{IgnDeg1}(x_3, x_4, x_1, x_2, z)
\end{split}\label{eq_bs_lra_ecu}
\end{align}
\end{example}

Terms of the two arithmetic sorts are constructed from a set $\varset$ of \emph{variables}, 
the set of integer constants $c\in\Int$, and binary function symbols $+$ and $-$ (written infix).
Atoms in $\BS(\LA)$ are either \emph{first-order atoms} (e.g., $\text{IgnTable}(0,13,880,1100,z)$) or
\emph{(linear) arithmetic atoms} (e.g., $x_2 < 880$).
Arithmetic atoms may use the predicates $\leq, <, \neq, =, >, \geq$, which are written infix and have the expected
fixed interpretation. 
Predicates used in first-order atoms are called \emph{free}.
\emph{First-order literals} and related notation is defined as before.
\emph{Arithmetic literals} coincide with arithmetic atoms, since
the arithmetic predicates are closed under negation, e.g., $\neg(x_2\geq 1100)\equiv x_2<1100$.

$\BS(\LA)$ clauses and conjectures are defined as for \BS{} but using $\BS(\LA)$ atoms.
We often write Horn clauses in the form $\Lambda \parallel \Delta \rightarrow H$ where $\Delta$ is a multiset of free first-order atoms, $H$ is either a first-order atom or $\bot$, and $\Lambda$ is a multiset of $\LA$ atoms. 
The semantics of a clause in the form $\Lambda \parallel \Delta \rightarrow H$ is $\bigvee_{\lambda \in \Lambda} \neg \lambda \vee \bigvee_{A \in \Delta} \neg A \vee H$, e.g.,
the clause $x >1 \lor y \neq 5 \lor \neg Q(x) \lor R(x, y)$ is also written $x\leq 1, y = 5 || Q(x) \rightarrow R(x,y)$.

A clause or clause set is \emph{abstracted} if its first-order literals contain only variables or first-order constants.
Every clause $C$ is equivalent to an abstracted clause that is obtained by replacing each non-variable arithmetic term $t$
that occurs in a first-order atom by a fresh variable $x$ while adding an arithmetic atom $x\neq t$ to $C$.
We asssume abstracted clauses for theory development, but we prefer non-abstracted
clauses in examples for readability,e.g., a fact $P(3,5)$ is considered in the development
of the theory as the clause $x=3, x=5 \parallel \rightarrow P(x,y)$, this is important when collecting the
necessary test points.
Moreover, we assume that all variables in the theory part of a clause also appear in the first order part, i.e., $\vars(\Lambda) \subseteq \vars(\Delta \rightarrow H)$ for every clause $\Lambda \parallel \Delta \rightarrow H$.
If this is not the case for $x$ in $\Lambda \parallel \Delta \rightarrow H$, 
then we can easily fix this by first introducing a fresh unary predicate $Q$ over the $\sort(x)$,
then adding the literal $Q(x)$ to $\Delta$, and finally adding a clause $\parallel \rightarrow Q(x)$ to our clause set.
Alternatively, $x$ could be eliminated by $\LA$ variable elimintation in our context, however this results in a worst case
exponential blow up in size.
This restriction is necessary because we base all our computations for the test-point scheme on predicate argument positions
and would not get any test points for variables that are not connected to any predicate argument positions.

\myparagraph{Simpler Forms of Linear Arithmetic:} \label{sec:bsslr}
The main logic studied in this paper is obtained by restricting $\HBS(\LA)$ to 
a simpler form of linear arithmetic. We first introduce a simpler logic $\HBS(\SB)$
as a well-known fragment of $\HBS(\LA)$ for which satisfiability is decidable~\cite{GeMoura09,HorbachEtAl17CADE},
and later present the generalization $\HBS(\LA)\PAG$ of this formalism that we will use.

\begin{definition}
The \emph{Horn Bernays-Schönfinkel fragment over simple linear arithmetic}, $\HBS(\SB)$, is a subset of
$\HBS(\LA)$ where all arithmetic atoms are of
the form $x \LAOP c$ or $d \LAOP c$, such that $c\in\Int$, $d$ is a (possibly free) constant, $x\in \varset$, and $\LAOP \in \{\leq, <, \neq, =, >, \geq\}$.
\end{definition}

Please note that $\HBS(\SB)$ clause sets may be unpure due to free first-order constants of an arithmetic sort. Studying
unpure fragments is beyond the scope of this paper but they show up in applications as well.

\begin{example}\label{ex_bs_lra_ecu_nonground}
The ECU use case leads to $\HBS(\LA)$ clauses such as
\begin{align}\begin{split}
  &x_1 < y_1 \;\lor\; x_1\geq y_2 \;\lor\; x_2 < y_3  \;\lor\; x_2\geq y_4 \;\lor\; \neg \text{KPa}(x_3, x_1) \;\lor{} \\
  &  \neg \text{Rpm}(x_4, x_2) \;\lor\; \neg \text{IgnTable}(y_1,y_2,y_3,y_4,z)  \;\lor\; \text{IgnDeg1}(x_3, x_4, x_1, x_2, z).
\end{split}\label{eq_bs_lra_ecu_nonground}
\end{align}
This clause is not in $\HBS(\SB)$, e.g., since $x_1 > x_5$ is not allowed in $\BS(\SB)$.
However, clause \eqref{eq_bs_lra_ecu} of Example~\ref{ex_bs_lra_ecu} is a $\BS(\SB)$ clause that is an instance of \eqref{eq_bs_lra_ecu_nonground},
obtained by the substitution $\{y_1\mapsto 0,y_2\mapsto 13,y_3\mapsto 880,y_4\mapsto 1100\}$. This grounding will eventually be obtained
by resolution on the $\text{IgnTable}$ predicate, because it occurs only positively in ground unit facts.
\end{example}

Example~\ref{ex_bs_lra_ecu_nonground} shows that $\HBS(\SB)$ clauses can sometimes be obtained by instantiation.
In fact, for the satisfiability of an $\HBS(\LA)$ clause set $N$ only those instances of clauses $(\Lambda \parallel \Delta \rightarrow H) \sigma$ are \emph{relevant}, 
for which we can actually derive all ground facts $A \in \Delta \sigma$ by resolution from $N$.
If $A$ cannot be derived from $N$ and $N$ is satisfiable, 
then there always exists a satisfying interpretation $\inta$ that interprets $A$ as false
(and thus $(\Lambda \parallel \Delta \rightarrow H) \sigma$ as true).
Moreover, if those relevant instances can be simplified to $\HBS(\SB)$ clauses, 
then it is possible to extend almost all $\HBS(\SB)$ techniques (including our Datalog hammer) to those $\HBS(\LA)$ clause sets.

In our case resolution means \emph{hierarchic unit resolution}:
given a clause $\Lambda_1 \parallel L, \Delta \rightarrow H$ and a unit clause $\Lambda_2 \parallel \rightarrow K$ with $\sigma = \mMGU(L,K)$,
their \emph{hierarchic resolvent} is $(\Lambda_1,\Lambda_2 \parallel \Delta \rightarrow H)\sigma$.
A fact $P(\bar{a})$ is \emph{derivable} from a pure set of $\HBS(\LA)$ clauses $N$ if there exists a
clause $\Lambda \parallel \rightarrow P(\bar{t})$ that (i)~is the result of a sequence of unit resolution steps from the clauses in $N$ and
(ii)~has a grounding $\sigma$ such that $P(\bar{t}) \sigma = P(\bar{a})$ and $\Lambda \sigma$ evaluates to true.
If $N$ is satisfiable, then this means that any fact $P(\bar{a})$ derivable from $N$ is true in all satisfiable interpretations of $N$, i.e., $N \models P(\bar{a})$.
We denote the \emph{set of derivable facts} for a predicate $P$ from $N$ by $\dfacts(P,N)$.
A \emph{refutation} is the sequence of resolution steps that produces a
clause $\Lambda \parallel \rightarrow \bot$ with $\sigval^{\LA} \models \Lambda\delta$ for some grounding $\delta$.
\emph{Hierarchic unit resolution} is sound and refutationally complete for pure $\HBS(\LA)$, 
since every set $N$ of pure $\HBS(\LA)$ clauses $N$ is \emph{sufficiently complete}~\cite{BachmairGanzingerEtAl94}, and hence
 \emph{hierarchic superposition} is sound and refutationally complete for $N$~\cite{BachmairGanzingerEtAl94,BaumgartnerWaldmann19}.
 
So naturally if all derivable facts of a predicate $P$ already appear in $N$, 
then only those instances of clauses can be relevant whose occurrences of $P$ match those facts (i.e., can be resolved with them).
We call predicates with this property positively grounded:

\begin{definition}[Positively Grounded Predicate~\cite{BrombergerEtAl21FROCOS}]
  Let $N$ be a set of $\HBS(\LA)$ clauses. A free first-order predicate $P$ is
  a \emph{positively grounded predicate} in $N$ if all positive occurrences of
  $P$ in $N$ are in ground unit clauses (also called facts).
\end{definition}

\begin{definition}[Positively Grounded $\HBS(\SB)$: $\HBS(\SB)\SP$~\cite{BrombergerEtAl21FROCOS}]
  An $\HBS(LA)$ clause set $N$ is out of the fragment \emph{positively grounded $\HBS(\SB)$ ($\HBS(\SB)\SP$)}\ if we can transform $N$ into an $\HBS(\SB)$ clause set $N'$ 
by first resolving away all negative occurrences of positively grounded predicates $P$ in $N$, 
simplifying the thus instantiated $\LA$ atoms,
and finally eliminating all clauses where those predicates occur negatively. 
\end{definition}

As mentioned before, if all relevant instances of an $\HBS(\LA)$ clause set can be simplified to $\HBS(\SB)$ clauses, 
then it is possible to extend almost all $\HBS(\SB)$ techniques (including our Datalog hammer) to those clause sets.
$\HBS(\SB)\SP$ clause sets have this property and this is the reason,
why we managed to extend our Datalog hammer to pure $\HBS(\SB)\SP$ clause sets in~\cite{BrombergerEtAl21FROCOS}.
For instance, the set $N = \{P(1), \; P(2), \; Q(0),\; (x \leq y + z \parallel P(y), Q(z) \rightarrow R(x,y))\}$ is an $\HBS(\LA)$ clause set,
but not an $\HBS(\SB)$ clause set due to the inequality $x \leq y + z$.
Note, however, that the predicates $P$ and $Q$ are positively grounded, the only positive occurrences of $P$ and $Q$ are the facts $P(1)$, $P(2)$, and $Q(0)$.
If we resolve with the facts for $P$ and $Q$ and simplify, 
then we get the clause set $N' = \{P(1), \; P(2), \; Q(0),\; (x \leq 1 \parallel \rightarrow R(x,1)), \; (x \leq 2 \parallel \rightarrow R(x,2))\}$, which does now belong to $\HBS(\SB)$.
This means $N$ is a positively grounded $\HBS(\SB)$ clause set and our Datalog hammer can still handle it.

Positively grounded predicates are only one way to filter out irrelevant clause instances.
As part of our improvements, we define in Section~\ref{sec:typing} a new logic called approximately
grounded $\HBS(\SB)$ ($\HBS(\SB)\PAG$) that is an extension of $\HBS(\SB)\SP$ and serves as the new input logic of our sorted Datalog hammer.
It is based on over-approximating the set of \emph{derivable values} $\dval(P,i,N) = \{a_i \mid P(\bar{a}) \in \dfacts(P,N)\}$ for each argument position $i$ of each predicate $P$ in $N$ 
with only finitely many derivable values, i.e., $|\dval(P,i,N)| \in \Nat$.
These argument positions are also called \emph{finite}. 
With regard to clause relevance, only those clause instances are relevant, where a finite argument position is instantiated by one of the derivable values.
We call a set of clauses $N$ an approximately grounded $\HBS(\SB)$ clause set if all 
relevant instances based on this criterion can be simplified to $\HBS(\SB)$ clauses.
For instance, the set $N = \{(x \leq 1 \parallel \rightarrow P(x,1)), \; (x > 2 \parallel \rightarrow P(x,3)), \; (x \geq 0 \parallel \rightarrow Q(x,0)), \; (u \leq y + z \parallel P(x,y), Q(x,z) \rightarrow R(x,y,z,u))\}$ is an $\HBS(\LA)$ clause set, but not a (positively grounded) $\HBS(\SB)$ clause set due to the inequality $z \leq y + u$ and the lack of positively grounded predicates.
However, the argument positions $(P,2)$, $(Q,2)$, $(R,2)$ and $(R,3)$ only have finitely many derivable values $\dval(P,2,N) = \dval(R,2,N) =\{1,3\}$ and $\dval(Q,2,N) = \dval(R,3,N) =\{0\}$.
If we instantiate all occurrences of $P$ and $Q$ over those values, then we get the set $N' = \{(x \leq 1 \parallel \rightarrow P(x,1)), \; (x > 2 \parallel \rightarrow P(x,3)), \; (x \geq 0 \parallel \rightarrow Q(x,0)), \; (u \leq 1 \parallel P(x,1), Q(x,0) \rightarrow R(x,1,0,u)), \; (u \leq 3 \parallel P(x,3), Q(x,0) \rightarrow R(x,3,0,u))\}$ that is an $\HBS(\SB)$ clause set.
This means $N$ is an approximately grounded $\HBS(\SB)$ clause set and our extended Datalog hammer can handle it.

\subsubsection{Test-Point Schemes and Functions}
The Datalog hammer in~\cite{BrombergerEtAl21FROCOS} is based on the following idea: 
For any pure $\HBS(\SB)$ clause set $N$ that is unsatisfiable, we only need to look at the instances $\mGnd_B(N)$ of $N$ over finitely many test points $B$ to construct a refutation.
Symmetrically, if $N$ is satisfiable, 
then we can extrapolate a satisfying interpretation for $N$ from a satisfying interpretation for $\mGnd_B(N)$.
If we can compute such a set of test points $B$ for a clause set $N$, 
then we can transform the clause set into an equisatisfiable Datalog program.
There exist similar properties for universal/existential conjectures.
A \emph{test-point scheme} is an algorithm that can compute such a set of test points $B$ for any $\HBS(\SB)$ clause set $N$
and any conjecture $N \models \mathcal{Q} \bar{x}. P(\bar{x})$ with $\mathcal{Q} \in \{{\exists},{\forall}\}$.

The test-point scheme used by our original Datalog hammer computes the same set of test points for all variables and predicate argument positions. 
This has several disadvantages:
(i)~it cannot handle variables with different sorts and
(ii)~it often selects too many test points (per argument position) because it cannot recognize which inequalities and which argument positions are connected.
The goal of this paper is to resolve these issues.
However, this also means that we have to assign different test-point sets to different predicate argument positions.
We do this with so-called test-point functions.

A \emph{test-point function} ($\tpfunction$) $\beta$ is a function that assigns to some argument positions $i$ of some predicates $P$ a set of test points $\beta(P,i)$.
An argument position $(P,i)$ is assigned a set of test points if $\beta(P,i) \subseteq \sort(P,i)^\inta$ and otherwise $\beta(P,i) = \bot$.
A test-point function $\beta$ is \emph{total} if all argument positions $(P,i)$ are assigned, i.e., $\beta(P,i) \neq \bot$.

A variable $x$ of a clause $\Lambda \parallel \Delta \rightarrow H$ occurs in an argument position $(P,i)$ if $(P,i) \in \depend(x,\Lambda \parallel \Delta \rightarrow H)$, 
where $\depend(x,Y) = \{(P,i) \mid P(\bar{t}) \in \atoms(Y) \text{ and } t_i = x\}$.
Similarly, a variable $x$ of an atom $Q(\bar{t})$ occurs in an argument position $(Q,i)$ if $(Q,i) \in \depend(x,Q(\bar{t}))$.
A substitution $\sigma$ for a clause $Y$ or atom $Y$ is a \emph{well-typed instance} over a $\tpfunction$ $\beta$ if it guarantees for each variable $x$ that $x \sigma$ is an element of $\sort(x)^{\inta}$ and part of every test-point set (i.e., $x \sigma \in \beta(P,i)$) of every argument position $(P,i)$ it occurs in (i.e., $(P,i)\in\depend(x,Y)$) and that is assigned a test-point set by $\beta$ (i.e., $\beta(P,i) \neq \bot$). 
To abbreviate this, we define a set $\wti(x,Y,\beta)$ that contains all values with which a variable can fulfill the above condition, i.e.,
$\wti(x,Y,\beta) = \sort(x)^{\inta} \cap (\bigcap_{(P,i) \in \depend(x,Y) \text{ and } \beta(P,i) \neq \bot} \beta(P,i))$.
Following this definition, we denote by $\tg_{\beta}(Y)$ the \emph{set of all well-typed instances} for a clause/atom $Y$ over the $\tpfunction$ $\beta$, or formally:
$\tg_{\beta}(Y) = \{ \sigma \mid \forall x \in \vars(Y). (x \sigma) \in \wti(x,Y,\beta)\}$.
With the function $\mGnd_{\beta}$, we denote the \emph{set of all well-typed ground instances} of a clause/atom $Y$ over the $\tpfunction$ $\beta$, i.e., 
$\mGnd_{\beta}(Y) = \{Y \sigma \mid \sigma \in \tg_{\beta}(Y)\}$, 
or a set of clauses $N$, i.e., $\mGnd_{\beta}(N) = \{Y \sigma \mid Y \in N \text{ and } \sigma \in \tg_{\beta}(Y)\}$.

The most general $\tpfunction$, denoted by $\beta^*$, assigns each argument position to the interpretation of its sort, i.e., $\beta^*(P,i)= \sort(P,i)^\inta$. 
So depending on the sort of $(P,i)$, either to $\Real$, $\Int$, or one of the $\Fol_i$.
A set of clauses $N$ is satisfiable if and only if $\mGnd_{\beta^*}(N)$, the set of all ground instances of $N$ over the base sorts, is satisfiable.
Since $\beta^*$ is the most general $\tpfunction$, we also write $\mGnd(Y)$ for $\mGnd_{\beta^*}(Y)$ and $\tg(Y)$ for $\tg_{\beta^*}(Y)$.

If we restrict ourselves to test points,
then we also only get interpretations over test points and not for the full base sorts.
In order to extrapolate an interpretation from test points to their full sorts,
we define extrapolation functions ($\epfunctions$) $\eta$. 
An \emph{extrapolation function} ($\epfunction$) $\eta(P,\bar{a})$ maps an argument vector of test points for predicate $P$ (with $a_i \in \beta(P,i)$) to the subset of
$\sort(P,1)^\inta \times \ldots \times \sort(P,n)^\inta$ that is supposed to be interpreted the same as $\bar{a}$, 
i.e., $P(\bar{a})$ is interpreted as true if and only if $P(\bar{b})$ with $\bar{b} \in \eta(P,\bar{a})$ is interpreted as true.
By default, any argument vector of test points $\bar{a}$ for $P$ must also be an element of $\eta(P,\bar{a})$, i.e., $\bar{a} \in \eta(P,\bar{a})$.
An extrapolation function does not have to be complete for all argument positions, i.e., 
there may exist argument positions from which we cannot extrapolate to all argument vectors. 
Formally this means that the actual set of values that can be extrapolated from $(P,i)$ (i.e., $\bigcup_{a_1 \in \beta(P,1)} \ldots \bigcup_{a_n \in \beta(P,n)} \eta(P,\bar{a})$) may be a strict subset of
$\sort(P,1)^\inta \times \ldots \times \sort(P,n)^\inta$.
For all other values $\bar{a}$, $P(\bar{a})$ is supposed to be interpreted as false.

\subsubsection{Covering Clause Sets and Conjectures}

Our goal is to create total $\tpfunctions$ that restrict our solution space from the infinite reals and integers to finite sets of test points while still preserving (un)satisfiability.
Based on these $\tpfunctions$, we are then able to define a Datalog hammer that transforms a clause set belonging to (an extension of) $\HBS(\LA)$ into an equisatisfiable $\HBS$ clause set; 
even modulo universal and existential conjectures.

To be more precise, we are interested in finite $\tpfunctions$ (together with matching $\epfunctions$) that cover a clause set $N$ or a conjecture $N \models \mathcal{Q} \bar{x}. P(\bar{x})$ with $\mathcal{Q} \in \{{\exists},{\forall}\}$.
A total $\tpfunction$ $\beta$ is \emph{finite} if each argument position is assigned to a finite set of test points, i.e., $|\beta(P,i)| \in \Nat$.
A $\tpfunction$ $\beta$ \emph{covers a set of clauses} $N$ if $\mGnd_{\beta}(N)$ is equisatisfiable to $N$.
A $\tpfunction$ $\beta$ \emph{covers a universal conjecture} $\forall \bar{x}. Q(\bar{x})$ over $N$ if $\mGnd_{\beta}(N) \cup N_Q$
is satisfiable if and only if $N \models \forall \bar{x}. Q(\bar{x})$ is false.
Here $N_Q$ is the set $\{\parallel \mGnd_{\beta}(Q(\bar{x})) \rightarrow \bot \}$ if $\eta$ is complete
for $Q$ 
or the empty set otherwise.
A $\tpfunction$ $\beta$ \emph{covers an existential conjecture} $N \models \exists \bar{x}. Q(\bar{x})$ if $\mGnd_{\beta}(N) \cup \mGnd_{\beta}(\parallel Q(\bar{x}) \rightarrow \bot)$ is satisfiable if and only if $N \models \exists \bar{x}. Q(\bar{x})$ is false.

The most general $\tpfunction$ $\beta^*$ obviously covers all \HBS(\LA) clause sets and conjectures because satisfiability of $N$ is defined over $\mGnd_{\beta^*}(N)$. 
However, $\beta^*$ is not finite.
The test-point scheme in~\cite{BrombergerEtAl21FROCOS}, which assigns one finite set
of test points $B$ to all variables, also covers clause sets and universal/existential conjectures; at least if we restrict our input to variables over the reals.
As mentioned before, the goal of this paper is improve this test-point scheme by 
assigning different test-point sets to different predicate argument positions.


\section{The Sorted Datalog Hammer}
\label{sec:typing}

In this section, we present a transformation that we call the \emph{sorted Datalog hammer}. 
It transforms any pure $\HBS(\SB)$ clause set modulo a conjecture into an $\HBS$ clause set.
To guide our explanations, 
we apply each step of the transformation to a simplified example of the 
electronic control unit use case:

\begin{example}\label{exp:runningecuexp}
An electronic control unit (ECU) of a combustion engine determines actuator operations.
For instance, it computes the ignition timings based on a set of input sensors.
To this end, it looks up some base factors from static tables and combines them to the actual actuator values through a series of rules.

In our simplified model of an ECU, 
we only compute one actuator value, the ignition timing, 
and we only have an engine speed sensor (measuring in Rpm) as our input sensor.
Our verification goal, expressed as a universal conjecture, is to confirm, 
that the ECU computes an ignition timing for all potential input sensor values. 
Determining completeness of a set of rules, i.e., determining that the rules produce a result for all potential input values, 
is also our most common application for universal conjectures.
The ECU model is encoded as the following pure $\HBS(\LA)$ clause set $N$:\newline
$D_1 : \SpeedTable(0,2000,1350), \quad D_2 : \SpeedTable(2000,4000,1600),$\newline
$D_3: \SpeedTable(4000,6000,1850), \quad D_4 : \text{\SpeedTable}(6000,8000,2100),$ \newline
$C_1 : 0 \leq x_p, x_p < 8000 \parallel \rightarrow \Speed(x_p),$\newline 
$C_2 : x_1 \leq x_p, x_p < x_2 \parallel \Speed(x_p), \SpeedTable(x_1,x_2,y) \rightarrow \IgnDeg(x_p,y), $\newline
$C_{3} : \IgnDeg(x_p,z) \rightarrow \ResDegArgs(x_p), \quad
C_{4} : \ResDegArgs(x_p) \rightarrow \Conjecture(x_p), $\newline
$C_{5} : x_p \geq 8000 \parallel \rightarrow \Conjecture(x_p), \quad
C_{6} : x_p < 0 \parallel \rightarrow \Conjecture(x_p), $\newline

In this example all variables are real variables.
The clauses $D_1 - D_4$ are table entries from which we determine the base factor of our ignition time based on the speed. 
Semantically, $D_1 : \SpeedTable(0,2000,1350)$ states that the base ignition time is $13.5^{\circ}$ before dead center if the engine speed lies between $0 \Rpm$ and $2000 \Rpm$.
The clause $C_1$ produces all possible input sensor values labeled by the predicate $\Speed$.
The clause $C_2$ determines the ignition timing from the current speed and the table entries.
The end result is stored in the predicate $\IgnDeg(x_p,z)$, 
where $z$ is the resulting ignition timing and $x_p$ is the speed that led to this result.
The clauses $C_3 - C_6$ are necessary for encoding the verification goal as a universal conjecture over a single atom.
In clause $C_3$, the return value is removed from the result predicate $\IgnDeg(x_p,z)$ because for the conjecture we only need to know that there is a result and not what the result is.
Clause $C_{4}$ guarantees that the conjecture predicate $\Conjecture(x_p)$ is true if the rules can produce a $\IgnDeg(x_p,z)$ for the sensor value.
Clauses $C_{5} \& C_{6}$ guarantee that the conjecture predicate is true if one of the sensor values is out of bounds.
This flattening process can be done automatically using the techniques outlined in~\cite{BrombergerEtAl21FROCOS}.
Hence, the ECU computes an ignition timing for all potential input sensor values if the universal conjecture $\forall x_p. \Conjecture(x_p)$ is entailed by $N$.
\end{example}

\subsubsection{Approximately Grounded}

Example~\ref{exp:runningecuexp} contains inequalities that go beyond simple variable bounds, e.g., $x_1 \leq x_p$ in $C_2$.
However, it is possible to reduce the example to an $\HBS(\SB)$ clause set.
As our first step of the sorted Datalog hammer, 
we explain a way to heuristically determine which $\HBS(\LA)$ clause sets can be reduced to $\HBS(\SB)$ clause sets.
Moreover, we show later that we do not have to explicitly perform this reduction but that we can extend our other algorithms to handle this heuristic extension of $\HBS(\SB)$ directly.

We start by formulating an extension of positively grounded $\HBS(\SB)$ called approximately grounded $\HBS(\SB)$.
It is based on an over-approximation of the derivable values for each argument position.
An argument position $(P,i)$ is \emph{(in)finite} if $\dfacts(P,N)$ contains (in)finitely many different values at position $i$, i.e., the set of \emph{derivable values} $\dval(P,i,N)$ is (in)finite.
Naturally, all argument positions over first-order sorts $\FolS$ are finite argument positions.
But other argument positions can be finite too, e.g., in the set of clauses $\{P(0), P(1), P(x)\rightarrow Q(x)\}$ all argument positions of all predicates are finite.

\begin{lemma}\label{lem:hbsfiniteishard}
Determining the finiteness of a predicate argument position in an $\HBS(\LA)$ clause set is undecidable.
\end{lemma}

Determining the finiteness of a predicate argument position (and all its derivable values) is not trivial. 
In general, it is as hard as determining the satisfiability of a clause set, 
so in the case of $\HBS(\LA)$ undecidable~\cite{Downey1972,HorbachEtAl17ARXIV}.
This is the reason, why we propose the following algorithm that over-approximates the derivable values for each predicate argument position to heuristically determine whether a set of $\HBS(\LA)$ clauses can be reduced to $\HBS(\SB)$.

\begin{center}
\begin{tabular}{l}
  $\text{DeriveValues}(N)$\\
  for all predicates $P$ and argument positions $i$ for $P$\\
  $\quad$ $\aval(P,i,N) := \emptyset$;\\
  $\text{change} := \top$;\\
  while (change) \\
  $\quad$ $\text{change} := \bot$;\\
  $\quad$ for all Horn clauses $\Lambda \parallel \Delta \rightarrow P(t_1,\ldots,t_n)\in N$\\
  $\quad$ $\quad$ for all argument positions $1\leq i\leq n$ where $\aval(P,i,N) \not= \Real$\\
  $\quad$ $\quad$ $\quad$ if $[(t_i = c)$ or $t_i$ is assigned a constant $c$ in $\Lambda$ and $c\not\in\aval(P,i,N)$] then\\
  $\quad$ $\quad$ $\quad$ $\quad$ $\aval(P,i,N) := \aval(P,i,N)\cup \{c\}, \text{change} := \top$;\\
  $\quad$ $\quad$ $\quad$ else if [$t_i$ appears in argument positions $(Q_1,k_1),\ldots,(Q_m,k_m)$ in $\Delta$\\
  $\quad$ $\quad$ $\quad$ \hspace*{5ex} and $\aval(P,i,N) \not\supseteq \bigcap_j \aval(Q_j,k_j,N)$ ] then\\
  $\quad$ $\quad$ $\quad$ $\quad$ if [$\Real \neq \bigcap_j \aval(Q_j,k_j,N)$] then\\
  $\quad$ $\quad$ $\quad$ $\quad$ $\quad$ $\aval(P,i,N) := \aval(P,i,N) \cup \bigcap_j \aval(Q_j,k_j,N), \text{change} := \top$;\\  
  $\quad$ $\quad$ $\quad$ $\quad$ else\\
  $\quad$ $\quad$ $\quad$ $\quad$ $\quad$ $\aval(P,i,N) := \Real, \text{change} := \top$;\\
\end{tabular}
\end{center}

At the start, $\text{DeriveValues}(N)$ sets $\aval(P,i,N) = \emptyset$ for all predicate argument positions. 
Then it repeats iterating over the clauses in $N$ and uses the current sets $\aval$ in order to derive new values, 
until it reaches a fixpoint.
Whenever, $\text{DeriveValues}(N)$ computes that a clause can derive infinitely many values for an argument position,
it simply sets $\aval(P,i,N) = \Real$ for both real and integer argument positions. 
This is the case, when we have a clause $\Lambda \parallel \Delta \rightarrow P(t_1,\ldots,t_n)$,
and an argument position $i$ for $P$, such that:
(i)~$t_i$ is not a constant (and therefore a variable),
(ii)~$t_i$ is not assigned a constant $c$ in $\Lambda$ (i.e., there is no equation $t_i = c$ in $\Lambda$),
(iii)~$t_i$ is only connected to argument positions $(Q_1,k_1),\ldots,(Q_m,k_m)$ in $\Delta$ 
that already have $\aval(Q_j,k_j,N) = \Real$.
The latter also includes the case that $t_i$ is not connected to any argument positions in $\Delta$.
For instance, $\text{DeriveValues}(N)$ would recognize that clause $C_1$ in example~\ref{exp:runningecuexp} can be used to derive infinitely many values for the argument position $(\Speed,1)$ because the variable $x_p$ is not assigned an equation in $C_1$'s theory constraint $\Lambda := (0 \leq x_p, x_p < 8000)$ and $x_p$ is not connected to any argument position on the left side of the implication.
Hence, $\text{DeriveValues}(N)$ would set $\aval(\Speed,1,N) = \Real$.

For each run through the while loop, at least one predicate argument position is set to $\Real$ or the set is extended by
at least one constant. The set of constants in $N$ as well as the number of
predicate argument positions in $N$ are finite, hence $\text{DeriveValues}(N)$ terminates. It is correct because in each step it over-approximates
the result of a hierarchic unit resulting resolution step, see Section~\ref{sec:prelims}.
The above algorithm is highly inefficient. 
In our own implementation, we only apply it if all clauses are non-recursive and by first ordering the clauses based on their dependencies.
This guarantees that every clause is visited at most once and is sufficient for both of our use cases.

We can use the $\aval$ to heuristically determine whether a set of clauses can be reduced to $\HBS(\SB)$.
Let $\aval(P,i,N) \supseteq \dval(P,i,N)$ be an over-approximation of the derivable values per argument position. 
Then $|\aval(P,i,N)| \in \Nat$ is a sufficient criterion for argument position $i$ of $P$ being finite in $N$.
Based on $\aval$, we can now build a $\tpfunction$ $\beta^a$ that maps all finite argument positions $(P,i)$ that our over-approximation detected to the over-approximation of their derivable values, i.e., $\beta^a(P,i) := \aval(P,i,N)$ if $|\aval(P,i,N)| \in \Nat$ and $\beta^a(P,i) := \bot$ otherwise.
With $\beta^a$ we derive the finitely grounded over-approximation $\aGnd(Y)$ of a set of clauses $Y$, a clause $Y$ or an atom $Y$.
This set is equivalent to $\mGnd_{\beta^a}(Y)$, except that we assume that all $\LA$ atoms are simplified until they contain at most one integer number and that $\LA$ atoms that can be evaluated are reduced to true and false and the respective clause simplified.
Based of $\aGnd(N)$ we define a new extension of $\HBS(\SB)$ called approximately grounded $\HBS(\SB)$:

\begin{definition}[Approximately Grounded $\HBS(\SB)$: $\HBS(\SB)\AG$]
  A clause set $N$ is out of the fragment \emph{approximately grounded $\HBS(\SB)$} or short \emph{$\HBS(\SB)\AG$}\ if $\aGnd(N)$
  is out of the $\HBS(\SB)$ fragment. It is called $\HBS(\SB)\PAG$ if it is also pure.
\end{definition}


\begin{example}
Executing $\text{DeriveValues}(N)$ on example~\ref{exp:runningecuexp} leads to the following results:\newline
$\aval(\SpeedTable,1,N) = \{0,2000,4000,6000\}$, \newline
$\aval(\SpeedTable,2,N) = \{2000,4000,6000,8000\}$, \newline
$\aval(\SpeedTable,3,N) = \{1350,1600,1850,2100\}$,\newline
$\aval(\IgnDeg,2,N) = \{1350,1600,1850,2100\}$, \newline
and all other argument positions $(P,i)$ are infinite so $\aval(P,i,N) = \mathbb{R}$ for them.

We can now easily check whether $\aGnd(N)$ would turn our clause set into an $\HBS(\SB)$ fragment by checking whether the following holds for all inequalities:
all variables in the inequality except for one must be connected to a finite argument position on the left side of the clause it appears in.
This guarantees that all but one variable will be instantiated in $\aGnd(N)$ and the inequality can therefore be simplified to a variable bound.
\end{example}

\subsubsection{Connecting Argument Positions and Selecting Test Points}
As our second step,
we are reducing the number of test points per predicate argument position 
by incorporating that not all argument positions are connected to all inequalities. 
This also means that we select different sets of test points for different argument positions.
For finite argument positions, 
we can simply pick $\aval(P,i,N)$ as its set of test points.
However, before we can compute the test-point sets for all other argument positions, 
we first have to determine to which inequalities and other argument positions they are connected.

Let $N$ be an $\HBS(\SB)\PAG$ clause set and $(P,i)$ an argument position for a predicate in $N$.
Then we denote by $\connectedArgs(P,i,N)$ the \emph{set of connected argument positions} and 
by $\connectedIneqs(P,i,N)$ the \emph{set of connected inequalities}.
Formally, $\connectedArgs(P,i,N)$ is defined as the minimal set that fulfills the following conditions: 
(i)~two argument positions $(P,i)$ and $(Q,j)$ are connected if they share a variable in a clause in $N$, i.e.,
$(Q,j) \in \connectedArgs(P,i,N)$ if $(\Lambda \parallel \Delta \rightarrow H) \in N$, $P(\bar{t}), Q(\bar{s}) \in \atoms(\Delta \cup \{H\})$, and $t_i = s_j = x$; and
(ii)~the connection relation is transitive, i.e., if $(Q,j) \in \connectedArgs(P,i,N)$, then $\connectedArgs(P,i,N) = \connectedArgs(Q,j,N)$.
Similarly, $\connectedIneqs(P,i,N)$ is defined as the minimal set that fulfills the following conditions: 
(i)~an argument position $(P,i)$ is connected to an instance $\lambda'$ of an inequality $\lambda$ if they share a variable in a clause in $N$, i.e.,  $\lambda' \in \connectedIneqs(P,i,N)$ if $(\Lambda \parallel \Delta \rightarrow H) \in N$, $P(\bar{t}) \in \atoms(\Delta \cup \{H\})$, $t_i = x$, $(\Lambda' \parallel \Delta' \rightarrow H') \in \aGnd(\Lambda \parallel \Delta \rightarrow H)$, $\lambda' \in \Lambda'$, and $\lambda' = x \LAOP c$ (where $\LAOP = \{<,>,\leq,\geq,=,\neq\}$ and $c \in \mathbb{Z}$); 
(ii)~an argument position $(P,i)$ is connected to a value $c \in \mathbb{Z}$ if $P(\bar{t})$ with $t_i = c$ appears in a clause in $N$, i.e.,  $(x = c) \in \connectedIneqs(P,i,N)$ if $(\Lambda \parallel \Delta \rightarrow H) \in N$, $P(\bar{t}) \in \atoms(\Delta \cup \{H\})$, and $t_i = c$;
(iii)~an argument position $(P,i)$ is connected to a value $c \in \mathbb{Z}$ if $(P,i)$ is finite and $c \in \aval(P,i,N)$, i.e.,  $(x = c) \in \connectedIneqs(P,i,N)$ if $(P,i)$ is finite and $c \in \aval(P,i,N)$; and 
(iv)~the connection relation is transitive, i.e., $\lambda \in \connectedArgs(Q,j,N)$ if $\lambda \in \connectedIneqs(P,i,N)$ and $(Q,j) \in \connectedArgs(P,i,N)$.

\begin{example}
To highlight the connections in example~\ref{exp:runningecuexp} more clearly, we use the same variable symbol for connected argument positions.
Therefore $(\SpeedTable,1)$ and $(\SpeedTable,2)$ are only connected to themselves and 
$\connectedArgs(\SpeedTable,3,N) = \{(\SpeedTable,3), (\IgnDeg,2)\}$, and 
$\connectedArgs(\Speed,1,N) = \{(\Speed,1),(\IgnDeg,1),$ $(\ResDegArgs,1),(\Conjecture,1)\}$,
Computing the connected argument positions is a little bit more complicated:
first, if a connected argument position is finite, then we have to add all values in $\aval$ as equations to the connected inequalities.
E.g., $\connectedIneqs(\SpeedTable,1,N) = \{x_1 = 0, x_1 = 2000, x_1 = 4000, x_1 = 6000\}$ because $\aval(\SpeedTable,1,N) = \{0, 2000, 4000, 6000\}$.
Second, we have to add all inequalities connected in $\aGnd(N)$.
Again this is possible without explicitly computing $\aGnd(N)$. 
E.g., for the inequality $x_1 \leq x_p$ in clause $C_2$, 
we determine that $x_1$ is connected to the finite argument position $(\SpeedTable,1)$ in $C_2$ and $x_p$ is not connected to any finite argument positions. 
Hence, we have to connect the following variable bounds to all argument positions connected to $x_p$, i.e., $\{x_1 \leq x_p \mid x_1 \in \aval(\SpeedTable,1,N) \} = \{x_p \geq 0, x_p \geq 2000, x_p \geq 4000, x_p \geq 6000\}$ to the argument positions $\connectedArgs(\Speed,1,N)$.
If we apply the above two steps to all clauses,
then we get as connected inequalities:
$\connectedIneqs(\SpeedTable,2,N) = \{x_2 = 2000, x_2 = 4000, x_3 = 6000, x_4 = 8000\}$, 
$\connectedIneqs(\SpeedTable,3,N) = \{y = 1350, y = 1600, y = 1850, y = 2100\}$, and
$\connectedIneqs(\Speed,1,N) = \{x_p < 0, x_p < 2000, x_p < 4000, x_p < 6000, x_p <  8000, x_p \geq 0, x_p \geq 2000, x_p \geq 4000, x_p \geq 6000, x_p \geq  8000\}$.
\end{example}

Now based on these sets we can construct a set of test points as follows:
For each argument position $(P,i)$, 
we partition the reals $\Real$ into intervals such that any variable bound in $\lambda \in \connectedIneqs(P,i,N)$ is satisfied by all points in one such interval $I$ or none. 
Since we are in the Horn case, this is enough to ensure that we derive facts \emph{uniformly} over those intervals and the integers/non-integers.
To be more precise, we derive facts \emph{uniformly} over those intervals and the integers because $P(\bar{a})$ is derivable from $N$ and $a_i \in I \cap \mathbb{Z}$ 
implies that $P(\bar{b})$ is also derivable from $N$, where $b_j = a_j$ for $i \neq j$ and $b_i \in I \cap \mathbb{Z}$. 
Similarly, we derive facts \emph{uniformly} over those intervals and the non-integers because $P(\bar{a})$ is derivable from $N$ and $a_i \in I \setminus \mathbb{Z}$ 
implies that $P(\bar{b})$ is also derivable from $N$, where $b_j = a_j$ for $i \neq j$ and $b_i \in I$. 
As a result, it is enough to pick (if possible) one integer and one non-integer test point per interval to cover the whole clause set.

Formally we compute the interval partition $\iPart(P,i,N)$ and the set of test points $\testpoints(P,i,N)$ as follows:
First we transform all variable bounds $\lambda \in \connectedIneqs(P,i,N)$ into interval borders.
A variable bound $x \LAOP c$ with $\LAOP \in \{\leq,<,>,\geq\}$ in $\connectedIneqs(P,i,N)$ is turned into two interval borders. 
One of them is the interval border implied by the bound itself and the other its negation, e.g., $x \geq 5$ results in the interval border \dblquotes{}$[5$\dblquotes{} and the interval border of the negation \dblquotes{}$5)$\dblquotes{}.
Likewise, we turn every variable bound $x \LAOP c$ with $\LAOP \in \{=,\neq\}$ into all four possible interval borders for $c$, i.e. \dblquotes{}$c)$\dblquotes{}, \dblquotes{}$[c$\dblquotes{}, \dblquotes{}$c]$\dblquotes{}, and \dblquotes{}$(c$\dblquotes{}.
The set of interval borders $\iBorders(P,i,N)$ is then defined as follows:
\[\begin{array}{l l}
\iBorders(P,i,N) = &\left\{\dblquotes{}c]\dblquotes{}, \dblquotes{}(c\dblquotes{} \mid x \triangleleft c \in \connectedIneqs(P,i,N) \text{ where } \triangleleft \in \{\leq,=,\neq,>\} \right\} \; \cup \\
 &\left\{\dblquotes{}c)\dblquotes{}, \dblquotes{}[c\dblquotes{} \mid x \triangleleft c \in \connectedIneqs(P,i,N) \text{ where } \triangleleft \in \{\geq,=,\neq,<\} \right\} \; 
\cup \; \{\dblquotes{}(-\infty\dblquotes{},\dblquotes{}\infty)\dblquotes{}\}
\end{array}\]

The interval partition $\iPart(P,i,N)$ can be constructed by sorting $\iBorders(P,i,N)$ in an ascending order such that 
we first order by the border value---i.e. $\delta < \epsilon$ if $\delta \in \{\dblquotes{}c)\dblquotes{}, \dblquotes{}[c\dblquotes{}, \dblquotes{}c]\dblquotes{}, \dblquotes{}(c\dblquotes{}\}$, $\epsilon \in \{\dblquotes{}d)\dblquotes{}, \dblquotes{}[d\dblquotes{}, \dblquotes{}d]\dblquotes{}, \dblquotes{}(d\dblquotes{}\}$, and $c < d$---and 
then by the border type---i.e. $\dblquotes{}c)\dblquotes{} < \dblquotes{}[c\dblquotes{} < \dblquotes{}c]\dblquotes{} < \dblquotes{}(c\dblquotes{}$.
The result is a sequence $[\ldots,\dblquotes{}\delta_l\dblquotes{},\dblquotes{}\delta_u\dblquotes{},\ldots]$, 
where we always have one lower border $\delta_l$, 
followed by one upper border $\delta_u$. 
We can guarantee that an upper border $\delta_u$ follows a lower border $\delta_l$ because $\iBorders(P,i,N)$ always contains \dblquotes{}$c)$\dblquotes{} together with \dblquotes{}$[c$\dblquotes{} and \dblquotes{}$c]$\dblquotes{} together with \dblquotes{}$(c$\dblquotes{} for $c \in \Int$, so 
always two consecutive upper and lower borders.
Together with \dblquotes{}$(-\infty$\dblquotes{} and \dblquotes{}$\infty)$\dblquotes{} this guarantees that the sorted $\iBorders(P,i,N)$ has the desired structure.
If we combine every two subsequent borders $\delta_l$, $\delta_u$ in our sorted sequence $[\ldots,\dblquotes{}\delta_l\dblquotes{},\dblquotes{}\delta_u\dblquotes{},\ldots]$, then we receive our partition of intervals $\iPart(P,i,N)$.
For instance, if $x < 5$ and $x = 0$ are the only variable bounds in $\connectedIneqs(P,i,N)$, then $\iBorders(P,i,N) = \{\dblquotes{}5)\dblquotes{},\dblquotes{}[5\dblquotes{},\dblquotes{}0)\dblquotes{},\dblquotes{}[0\dblquotes{},\dblquotes{}0]\dblquotes{},\dblquotes{}(0\dblquotes{},\dblquotes{}(-\infty\dblquotes{},\dblquotes{}\infty)\dblquotes{}\}$ and if we sort and combine them we get $\iPart(P,i,N)=\{(-\infty,0),[0,0],(0,5),[5,\infty)\}$.

After constructing $\iPart(P,i,N)$, we can finally construct the set of test points $\testpoints(P,i,N)$ for argument position $(P,i)$.
If $|\aval(P,i,N)| \in \mathbb{N}$, i.e., we determined that $(P,i)$ is finite, 
then $\testpoints(P,i,N) = \aval(P,i,N)$.
If the argument position $(P,i)$ is over a first-order sort $\FolS_i$, i.e., $\sort(P,i)=\FolS_i$, 
then we should always be able to determine that $(P,i)$ is finite because $\Fol_i$ is finite. 
If the argument position $(P,i)$ is over an arithmetic sort, i.e., $\sort(P,i)=\RealS$ or $\sort(P,i)=\IntS$, 
and our approximation could not determine that $(P,i)$ is finite,
then the test-point set $\testpoints(P,i,N)$ for $(P,i)$ consists of at most two points per interval $I \in \iPart(P,i,N)$:
one integer value $a_I \in I \cap \mathbb{Z}$ if $I$ contains integers (i.e. if $I \cap \mathbb{Z} \neq \emptyset$) and one non-integer value $b_I \in I \setminus \mathbb{Z}$ if $I$ contains non-integers (i.e. if $I$ is not just one integer point).
Additionally, we enforce that $\testpoints(P,i,N) = \testpoints(Q,j,N)$ if $\connectedArgs(P,i,N) = \connectedArgs(Q,j,N)$ and both $(P,i)$ and $(Q,j)$ are infinite argument positions.
(In our implementation of this test-point scheme, we optimize the test point selection even further by picking only one test point per interval---if possible an integer value and otherwise a non-integer---if all $\connectedArgs(P,i,N)$ and all variables $x$ connecting them in $N$ have the same sort.
However, we do not prove this optimization explicitly here because the proofs are almost identical to the case for two test points per interval.)

Based on these sets, we can now also define a $\tpfunction$ $\beta$ and an $\epfunction$ $\eta$.
For the $\tpfunction$, we simply assign any argument position to $\testpoints(P,i,N)$, i.e., 
$\beta(P,i) = \testpoints(P,i,N) \cap \sort(P,i)^{\inta}$. 
(The intersection with $\sort(P,i)^{\inta}$ is needed to guarantee that the test-point set of an integer argument position is well-typed.)
This also means that $\beta$ is total and finite. 
For the $\epfunction$ $\eta$, we extrapolate any test-point vector $\bar{a}$ (with $\bar{a} = \bar{x} \sigma$ and $\sigma \in \tg_{\beta}(P(\bar{x}))$) over the (non-)integer subset of the intervals the test points belong to, i.e.,
$\eta(P,\bar{a}) = I'_1 \times \ldots \times I'_n$, where $I'_i  = \{a_i\}$ if we determined that $(P,i)$ is finite and otherwise $I_i$ is the interval $I_i \in \iPart(P,i,N)$ with $a_i \in I_i$ and $I'_i = I_i \cap \mathbb{Z}$ if $a_i$ is an integer value and $I'_i = I_i \setminus \mathbb{Z}$ if $a_i$ is a non-integer value.
Note that this means that $\eta$ might not be complete for every predicate $P$, e.g., 
when $P$ has a finite argument position $(P,i)$ with an infinite domain.
However, both $\beta$ and $\eta$ together still cover the clause set $N$, cover any universal conjecture $N \models \forall \bar{x}. Q(\bar{x})$, and cover any existential conjecture $N \models \exists \bar{x}. Q(\bar{x})$.

\begin{theorem}
\label{thm:coveringdefinition}
The $\tpfunction$ $\beta$ covers $N$. 
The $\tpfunction$ $\beta$ covers an existential conjecture $N \models \exists \bar{x}. Q(\bar{x})$. 
The $\tpfunction$ $\beta$ covers a universal conjecture $N \models \forall \bar{x}. Q(\bar{x})$.
\end{theorem}

\begin{example}
Continuation of example~\ref{exp:runningecuexp}:
The majority of argument positions in our example are finite. 
Hence, determining their test point set is equivalent to the over-approximation of derivable values $\aval$ we computed for them:
$\beta(\SpeedTable,1) = \{0,2000,4000,6000\}$,
$\beta(\SpeedTable,2) = \{2000,4000,6000,8000\}$,
$\beta(\SpeedTable,3) = \{1350,1600,1850,2100\}$, and
$\beta(\IgnDeg,2) = \{1350,1600,1850,2100\}$.
The other argument positions are all connected to $(\Speed,1)$ and 
$\connectedIneqs(\Speed,1,N) = \{x_p < 0, x_p < 2000, x_p < 4000, x_p < 6000, x_p <  8000, x_p \geq 0, x_p \geq 2000, x_p \geq 4000, x_p \geq 6000, x_p \geq  8000\}$, 
from which we can compute\newline
\centerline{$\iPart(P,i,N) = \{(-\infty,0),[0,2000),[2000,4000),[4000,6000),[6000,8000),[8000,\infty)\}$}
and select the test point sets
$\beta(\Speed,1) = $ $\beta(\IgnDeg,1) =$ $\beta(\ResDegArgs,1) =$ $\beta(\Conjecture,1) =$ $\{-1, 0, 2000, 4000, 6000, 8000\}$.
(Note that all variables in our problem are over the reals, so we only have to select one test point per interval! 
Moreover, in our previous version of the test point scheme, there would have been more intervals in the partition 
because we would have processed all inequalities, e.g., also those in $\connectedIneqs(\SpeedTable,3,N)$.)
The $\epfunction$ $\eta$ that determines which interval is represented by which test point is 
$\eta(P,1,-1) = (-\infty,0)$, $\eta(P,1,0) = [0,2000)$, $\eta(P,1,2000) = [2000,4000)$, $\eta(P,1,4000) = [4000,6000)$, $\eta(P,1,6000) = [6000,8000)$, $\eta(P,1,8000) =[8000,\infty)$
for the predicates $\Speed$, $\IgnDeg$, $\ResDegArgs$, and $\Conjecture$.
$\eta$ behaves like the identity function for all other argument positions because they are finite.
\end{example}

\subsubsection{From a Test-Point Function to a Datalog Hammer}
We can use the covering definitions, e.g., $\mGnd_{\beta}(N)$ is equisatisfiable to $N$, 
to instantiate our clause set (and conjectures) with numbers. 
As a result, we can simply evaluate all theory atoms and thus reduce our $\HBS(\SB)\PAG$ clause sets/conjectures to ground $\HBS$ clause sets, which means we could reduce our input into formulas without any arithmetic theory that can be solved by any Datalog reasoner.
There is, however, one problem.
The set $\mGnd_{\beta}(N)$ grows exponentially with regard to the maximum number of variables $n_C$ in any clause in $N$, i.e. $O(|\mGnd_{\beta}(N)|) = O(|N|\cdot|B|^{n_C})$,
where $B = \max_{(P,i)}(\beta(P,i))$ is the largest test-point set for any argument position.
Since $n_C$ is large for realistic examples, e.g., in our examples the size of $n_C$ ranges from 9 to 11 variables, the finite abstraction is often too large to be solvable in reasonable time. 
Due to this blow-up, we have chosen an alternative approach for our Datalog hammer.
This hammer exploits the ideas behind the covering definitions and will allow us to make the same ground deductions, but instead of grounding everything, we only need to (i)~ground the negated conjecture over our $\tpfunction$ and (ii)~provide a set of ground facts that define which theory atoms are satisfied by our test points.
As a result, the hammered formula is much more concise and we need no actual theory reasoning to solve the formula.
In fact, we can solve the hammered formula by greedily applying unit resolution until this produces the empty clause---which would mean the conjecture is implied---or until it produces no more new facts---which would mean we have found a counter example.
In practice, greedily applying resolution is not the best strategy and we recommend to use more advanced $\HBS$ techniques for instance those used by a state-of-the-art Datalog reasoner.

The Datalog hammer takes as input 
(i)~an $\HBS(\SB)\PAG$ clause set $N$ and 
(ii)~optionally a universal conjecture $\forall \bar{y}. P(\bar{y})$.
The case for existential conjectures is handled by encoding the conjecture $N \models \exists \bar{x}. Q(\bar{x})$ as the clause set $N \cup \{Q(\bar{x}) \rightarrow \bot\}$, which is unsatisfiable if and only if the conjecture holds.
Given this input, the Datalog hammer first computes the $\tpfunction$ $\beta$ and the $\epfunction$ $\eta$ as described above.
Next, it computes four clause sets that will make up the Datalog formula. 
The first set $\tren_N(N)$ is computed by abstracting away any arithmetic from the clauses $(\Lambda \parallel \Delta \rightarrow H) \in N$.
This is done by replacing each theory atom $A$ in $\Lambda$ with a literal $P_A(\bar{x})$, where $\vars(A) = \vars(\bar{x})$ and $P_A$
is a fresh predicate.
The abstraction of the theory atoms is necessary because Datalog does not support non-constant function symbols (e.g., $+,-$) that would otherwise appear in approximately grounded theory atoms.
Moreover, it is necessary to add extra sort literals $\neg Q_{(P,i,S)}(x)$ for some of the variables $x \in \vars(H)$, where $H = P(\bar{t})$, $t_i = x$, $\sort(x) = S$, and $Q_{(P,i,S)}$ is a fresh predicate.
This is necessary in order to define the test point set for $x$ if $x$ does not appear in $\Lambda$ or in $\Delta$.
It is also necessary in order to filter out any test points that are not integer values if $x$ is an integer variable (i.e. $\sort(x) = \IntS$) but
connected only to real sorted argument positions in $\Delta$ (i.e. $\sort(Q,j) = \RealS$ for all $(Q,j) \in \depend(x,\Delta)$).
It is possible to reduce the number of fresh predicates needed, e.g., by reusing the same predicate for two theory atoms whose variables range over the same sets of test points.
The resulting abstracted clause has then the form $\Delta_T, \Delta_S, \Delta \rightarrow H$, where $\Delta_T$ contains the abstracted theory literals (e.g. $P_A(\bar{x}) \in \Delta_T$) and $\Delta_S$ the ``sort'' literals (e.g. $Q_{(P,i,S)}(x) \in \Delta_S$).
The second set is denoted by $N_C$ and it is empty if we have no universal conjecture or if $\eta$ does not cover our conjecture.
Otherwise, $N_C$ contains the ground and negated version $\phi$ of our universal conjecture $\forall \bar{y}. P(\bar{y})$ . $\phi$ has the form $\Delta_{\phi} \rightarrow \bot$, where $\Delta_{\phi} = \mGnd_{\beta}(P(\bar{y}))$ contains all literals $P(\bar{y})$ for all groundings over $\beta$.
We cannot skip this grounding but the worst-case size of $\Delta_\phi$ is $O(\mGnd_{\beta}(P(\bar{y}))) = O(|B|^{n_{\phi}})$,
where $n_{\phi} = |\bar{y}|$, which is in our applications typically much smaller than the maximum number of variables $n_C$ contained in some clause in $N$.
The third set is denoted by $\tfacts(N,\beta)$ and contains a fact $\tren_N(A)$ for every ground theory atom $A$ contained in the theory part $\Lambda$ of a clause $(\Lambda \parallel \Delta \rightarrow H) \in \mGnd_{\beta}(N)$ such that $A$ simplifies to true. 
This is enough to ensure that our abstracted theory predicates evaluate every test point in every satisfiable interpretation $\inta$ to true that also would have evaluated to true in the actual theory atom.
Alternatively, it is also possible to use a set of axioms and a smaller set of facts and let the Datalog reasoner compute all relevant theory facts for itself.
The set $\tfacts(N,\beta)$ can be computed without computing $\mGnd_{\beta}(N)$ if we simply iterate over all theory atoms $A$ in all constraints $\Lambda$ of all clauses $Y = \Lambda \parallel \Delta \rightarrow H$ (with $Y \in N$) and compute all well typed groundings $\tau \in \tg_{\beta}(Y)$ such that $A \tau$ simplifies to true. 
This can be done in time $O(\mu(n_v) \cdot n_L \cdot {|B|}^{n_v})$ and the resulting set $\tfacts(N,\beta)$ has worst-case size $O(n_A \cdot {|B|}^{n_v})$,
where $n_L$ is the number of literals in $N$, $n_v$ is the maximum number of variables $|\vars(A)|$ in any theory atom $A$ in $N$, $n_A$ is the number of different theory atoms in $N$, and $\mu(x)$ is the time needed to simplify a theory atom over $x$ variables to a variable bound.
The last set is denoted by $\sfacts(N,\beta)$ and contains a fact $Q_{(P,i,S)}(a)$ for every fresh sort predicate $Q_{(P,i,S)}$ added during abstraction and every $a \in \beta(P,i) \cap S^\inta$. 
This is enough to ensure that $Q_{(P,i,S)}$ evaluates to true for every test point assigned to the argument position $(P,i)$ filtered by the sort $S$.
Please note that already satifiability testing for $\BS$ clause sets is NEXPTIME-complete in general, and DEXPTIME-complete for the Horn case~\cite{Lewis80,Plaisted84}.
So when abstracting to a polynomially decidable clause set (ground $\HBS$) an exponential factor is unavoidable.

\begin{lemma}\label{lem:hammeringequivalence}
$N$ is equisatisfiable to its hammered version $\tren_N(N) \cup \tfacts(N,\beta) \cup \sfacts(N,\beta)$.
The conjecture $N \models \exists \bar{y}. Q(\bar{y})$ is false iff $N_D = \tren_N'(N') \cup \tfacts(N',\beta) \cup \sfacts(N',\beta)$ is satisfiable with $N' = N \cup \{Q(\bar{y}) \rightarrow \bot\}$. 
The conjecture $N \models \forall \bar{y}. Q(\bar{y})$ is false iff $N_D = \tren_N(N) \cup \tfacts(N,\beta) \cup \sfacts(N,\beta) \cup N_C$ is satisfiable. 
\end{lemma}

Note that $\tren_N(N) \cup \tfacts(N,\beta) \cup \sfacts(N,\beta) \cup N_C$ is only a $\HBS$ clause set over a finite set of constants and not yet a Datalog input file.
It is well known that such a formula can be transformed easily into a Datalog problem by adding a nullary predicate Goal and adding it as a positive literal to any clause without a positive literal.
Querying for the Goal atom returns true if the $\HBS$ clause set was unsatisfiable and false otherwise.

\begin{example}
The hammered formula for example~\ref{exp:runningecuexp} looks as follows.
The set of renamed clauses $\tren_N(N)$ consists of all the previous clauses in $N$, 
except that inequalities have been abstracted to new first-order predicates:\newline
$D'_1 : \SpeedTable(0,2000,1350), \quad D'_2 : \SpeedTable(2000,4000,1600),$\newline
$D'_3: \SpeedTable(4000,6000,1850), \quad D'_4 : \text{\SpeedTable}(6000,8000,2100),$ \newline
$C'_1 : P_{0 \leq x_p}(x_p), P_{x_p < 8000}(x_p) \rightarrow \Speed(x_p),$\newline 
$C'_2 : P_{x_1 \leq x_p}(x_1,x_p), P_{x_p < x_2}(x_p,x_2), \Speed(x_p), \SpeedTable(x_1,x_2,y) \rightarrow \IgnDeg(x_p,y), $\newline
$C'_{3} : \IgnDeg(x_p,z) \rightarrow \ResDegArgs(x_p), \quad
C'_{4} : \ResDegArgs(x_p) \rightarrow \Conjecture(x_p), $\newline
$C'_{5} : P_{x_p \geq 8000}(x_p) \rightarrow \Conjecture(x_p), \quad
C'_{6} : P_{x_p < 0}(x_p) \rightarrow \Conjecture(x_p), $\newline
The set $\tfacts(N,\beta)$ defines for which test points those new predicates evaluate to true:
$\{P_{0 \leq x_p}(0)$, $P_{0 \leq x_p}(2000)$, $P_{0 \leq x_p}(4000)$, $P_{0 \leq x_p}(6000)$, $P_{0 \leq x_p}(8000)$, 
$P_{x_p < 8000}(-1)$,\newline $P_{x_p < 8000}(0)$, $P_{x_p < 8000}(2000)$, $P_{x_p < 8000}(4000)$, $P_{x_p < 8000}(6000)$, 
$P_{x_1 \leq x_p}(0,0)$, \newline $P_{x_1 \leq x_p}(0,2000)$, $P_{x_1 \leq x_p}(0,4000)$, $P_{x_1 \leq x_p}(0,6000)$, $P_{x_1 \leq x_p}(0,8000)$,
$P_{x_1 \leq x_p}(2000,2000)$, \newline $P_{x_1 \leq x_p}(2000,4000)$, $P_{x_1 \leq x_p}(2000,6000)$, $P_{x_1 \leq x_p}(2000,8000)$,
$P_{x_1 \leq x_p}(4000,4000)$, \newline$P_{x_1 \leq x_p}(4000,6000)$, $P_{x_1 \leq x_p}(4000,8000)$,
$P_{x_1 \leq x_p}(6000,6000)$, $P_{x_1 \leq x_p}(6000,8000)$,\newline
$P_{x_p < x_2}(-1,2000)$, $P_{x_p < x_2}(0,2000)$,
$P_{x_p < x_2}(-1,4000)$, $P_{x_p < x_2}(0,4000)$,\newline $P_{x_p < x_2}(2000,4000)$,
$P_{x_p < x_2}(-1,6000)$, $P_{x_p < x_2}(0,6000)$, $P_{x_p < x_2}(2000,6000)$, \newline $P_{x_p < x_2}(4000,6000)$,
$P_{x_p < x_2}(-1,8000)$, $P_{x_p < x_2}(0,8000)$, $P_{x_p < x_2}(2000,8000)$, \newline$P_{x_p < x_2}(4000,8000)$, $P_{x_p < x_2}(6000,8000)$,
$P_{x_p \geq 800}(8000)$,
$P_{x_p < 0}(-1)\}$ \newline
$\sfacts(N,\beta) = \emptyset$ because there are no fresh sort predicates.
The hammered negated conjecture is $N_C := \Conjecture(-1)$, $\Conjecture(0)$, $\Conjecture(2000)$, $\Conjecture(4000)$, $\Conjecture(6000)$, $\Conjecture(8000)$ $\rightarrow \bot$
and lets us derive false if and only if we can derive $\Conjecture(a)$ for all test points $a \in \beta(\Conjecture,1)$.
\end{example}

\section{Implementation and Experiments} \label{sec:experiments}

We have implemented the sorted Datalog hammer as an extension to the SPASS-SPL system~\cite{BrombergerEtAl21FROCOS} (option \verb|-d|) (SSPL in the table). 
By default the resulting formula is then solved with the Datalog reasoner VLog.
The previously file-based combination with the Datalog reasoner VLog has been replaced by an integration of VLog into SPASS-SPL via the VLog API.
We focus here only on the sorted extension and refer to~\cite{BrombergerEtAl21FROCOS} for an introduction into coupling of the two reasoners.
Note that the sorted Datalog hammer itself is not fine tuned towards the capabilities of a specific Datalog 
reasoner nor VLog towards the sorted Datalog hammer.

%

\begin{figure}[t]
\centering
  {\scriptsize
  \setlength{\tabcolsep}{1pt}
    \begin{tabular}{|l|c|c|c|c||r|r|r||r|r|r||r|r|r|r|}
    \hline
        \multicolumn{1}{|c|}{Problem} & \multicolumn{1}{|c|}{Q} & \multicolumn{1}{|c||}{Status} & \multicolumn{1}{|c|}{$|N|$} & \multicolumn{1}{|c|}{vars} & \multicolumn{1}{|c|}{$|B^m|$} & \multicolumn{1}{|c|}{$|\Delta_{\phi}|$} & \multicolumn{1}{|c||}{SSPL} & \multicolumn{1}{|c|}{$|B^s|$} & \multicolumn{1}{|c|}{$|\Delta^o_{\phi}|$} & \multicolumn{1}{|c||}{SSPL06} & \multicolumn{1}{|c|}{vampire} & \multicolumn{1}{|c|}{spacer} & \multicolumn{1}{|c|}{z3} & \multicolumn{1}{|c|}{cvc4} \\ \hline \hline
        lc\_e1 & $\exists$ & true & 139 & 9 & 9 & 0 & \textbf{< 0.1s} & 45 & 0 & \textbf{< 0.1s} & \textbf{< 0.1s} & \textbf{< 0.1s} & 0,1 & \textbf{< 0.1s} \\ \hline
        lc\_e2 & $\exists$ & false & 144 & 9 & 9 & 0 & \textbf{< 0.1s} & 41 & 0 & \textbf{< 0.1s} & \textbf{< 0.1s} & \textbf{< 0.1s} & - & - \\ \hline
        lc\_e3 & $\exists$ & false & 138 & 9 & 9 & 0 & \textbf{< 0.1s} & 37 & 0 & \textbf{< 0.1s} & \textbf{< 0.1s} & \textbf{< 0.1s} & - & - \\ \hline
        lc\_e4 & $\exists$ & true & 137 & 9 & 9 & 0 & \textbf{< 0.1s} & 49 & 0 & \textbf{< 0.1s} & \textbf{< 0.1s} & \textbf{< 0.1s} & \textbf{< 0.1s} & \textbf{< 0.1s} \\ \hline
        lc\_e5 & $\exists$ & false & 152 & 13 & 9 & 0 & 33.5s & - & - & N/A & \textbf{< 0.1s} & - & - & - \\ \hline
        lc\_e6 & $\exists$ & true & 141 & 13 & 9 & 0 & 42.8s & - & - & N/A & \textbf{0.1s} & 3.3s & 11.5s & 0.4s \\ \hline
        lc\_e7 & $\exists$ & false & 141 & 13 & 9 & 0 & 41.4s & - & - & N/A & \textbf{< 0.1s} & 7.6s & - & - \\ \hline
        lc\_e8 & $\exists$ & false & 141 & 13 & 9 & 0 & 32.5s & - & - & N/A & \textbf{< 0.1s} & 2.1s & - & - \\ \hline
        lc\_u1 & $\forall$ & false & 139 & 9 & 9 & 27 & \textbf{< 0.1s} & 45 & 27 & \textbf{< 0.1s} & \textbf{< 0.1s} & N/A & - & - \\ \hline
        lc\_u2 & $\forall$ & false & 144 & 9 & 9 & 27 & \textbf{< 0.1s} & 41 & 27 & \textbf{< 0.1s} & \textbf{< 0.1s} & N/A & - & - \\ \hline
        lc\_u3 & $\forall$ & true & 138 & 9 & 9 & 27 & \textbf{< 0.1s} & 37 & 27 & \textbf{< 0.1s} & \textbf{< 0.1s} & N/A & \textbf{< 0.1s} & \textbf{< 0.1s} \\ \hline
        lc\_u4 & $\forall$ & false & 137 & 9 & 9 & 27 & \textbf{< 0.1s} & 49 & 27 & \textbf{< 0.1s} & \textbf{< 0.1s} & N/A & - & - \\ \hline
        lc\_u5 & $\forall$ & false & 154 & 13 & 9 & 3888 & 32.4s & - & - & N/A & \textbf{0.1s} & N/A & - & - \\ \hline
        lc\_u6 & $\forall$ & true & 154 & 13 & 9 & 3888 & 32.5s & - & - & N/A & \textbf{2.3s} & N/A & - & - \\ \hline
        lc\_u7 & $\forall$ & true & 141 & 13 & 9 & 972 & 32.3s & - & - & N/A & \textbf{0.2s} & N/A & - & - \\ \hline
        lc\_u8 & $\forall$ & false & 141 & 13 & 9 & 1259712 & \textbf{48.8s} & - & - & N/A & 2351.4s & N/A & - & - \\ \hline
        ecu\_e1 & $\exists$ & false & 757 & 10 & 96 & 0 & \textbf{< 0.1s} & 624 & 0 & 1.3s & 0.2s & 0.1s & - & - \\ \hline
        ecu\_e2 & $\exists$ & true & 757 & 10 & 96 & 0 & \textbf{< 0.1s} & 624 & 0 & 1.3s & 0.2s & 0.1s & 1.4s & 0.4s \\ \hline
        ecu\_e3 & $\exists$ & false & 775 & 11 & 196 & 0 & 50.1s & 660 & 0 & 41.5s & 3.1s & \textbf{0.1s} & - & - \\ \hline
        ecu\_u1 & $\forall$ & true & 756 & 11 & 96 & 37 & \textbf{0.1s} & 620 & 306 & 1.1s & 32.8s & N/A & 197.5s & 0.4s \\ \hline
        ecu\_u2 & $\forall$ & false & 756 & 11 & 96 & 38 & \textbf{0.1s} & 620 & 307 & 1.1s & 32.8s & N/A & - & - \\ \hline
        ecu\_u3 & $\forall$ & true & 745 & 9 & 88 & 760 & \textbf{< 0.1s} & 576 & 11360 & 0.7s & 1.2s & N/A & 239.5s & 0.1s \\ \hline
        ecu\_u4 & $\forall$ & true & 745 & 9 & 486 & 760 & \textbf{< 0.1s} & 2144 & 237096 & 15.9s & 1.2s & N/A & 196.0s & 0.1s \\ \hline
        ecu\_u5 & $\forall$ & true & 767 & 10 & 96 & 3900 & \textbf{0.1s} & 628 & 415296 & 31.9s & - & N/A & - & - \\ \hline
        ecu\_u6 & $\forall$ & false & 755 & 10 & 95 & 3120 & \textbf{< 0.1s} & 616 & 363584 & 14.4s & 597.8 & N/A & - & - \\ \hline
        ecu\_u7 & $\forall$ & false & 774 & 11 & 196 & 8400 & \textbf{48.9s} & 656 & 2004708 & - & - & N/A & - & - \\ \hline
        ecu\_u8 & $\forall$ & true & 774 & 11 & 196 & 8400 & \textbf{48.7s} & 656 & 2004708 & - & - & N/A & - & - \\ \hline 
    \end{tabular}
  }
  \caption{Benchmark results and statistics}
  \label{fig:toolchainresults}
\end{figure}

In order to test the progress in efficiency of our sorted hammer, 
we ran the benchmarks of the lane change assistant and engine ECU  from~\cite{BrombergerEtAl21FROCOS} plus
more sophisticated, extended formalizations. While for the ECU benchmarks in~\cite{BrombergerEtAl21FROCOS} we modeled
ignition timing computation adjusted by inlet temperature measurements, the new benchmarks take also
gear box protection mechanisms into account. 
The lane change examples in~\cite{BrombergerEtAl21FROCOS} 
only simulated the supervisor for lane change assistants over some real-world instances.
The new lane change benchmarks check properties for all potential inputs.
The universal ones check that any suggested action by a lane change assistant is either proven as correct or disproven by our supervisor.
The existential ones check safety properties, e.g., that the supervisor never returns both a proof and a disproof for the same input.
We actually used SPASS-SPL to debug a prototype supervisor for lane change assistants during its development.
The new lane change examples are based on versions generated during this debugging process where SPASS-SPL found the following bugs:
(i)~it did not always return a result, 
(ii)~it declared actions as both safe and unsafe at the same time, and 
(iii)~it declared actions as safe although they would lead to collisions.
The supervisor is now fully verified.

The names of the problems are formatted so the lane change examples start with lc and the ECU examples start with ecu.
Our benchmarks are prototypical for the complexity of $\HBS(\SB)$ reasoning in that they cover all abstract relationships between
conjectures and $\HBS(\SB)$ clause sets. With respect to our two case studies we have many more examples showing respective characteristics.
We would have liked to run benchmarks from other sources, 
but could not find any problems in the SMT-LIB~\cite{BarrettMRST10,SMT-LIB} or CHC-COMP~\cite{CHC-COMP} benchmarks within the range of what our hammer can currently accept.
Either the arithmetic part goes beyond $\SB$ or there are further theories involved such as equality on first-order symbols.

For comparison, we also tested several state-of-the-art theorem provers for related logics (with the best settings we found): SPASS-SPL-v0.6 (SSPL06 in the table) that uses the original version of our Datalog Hammer~\cite{BrombergerEtAl21FROCOS} with settings \verb|-d| for existential and \verb|-d -n| for universal conjectures;
the satisfiability modulo theories (SMT) solver \emph{cvc4-1.8}~\cite{BarrettCDHJKRT:11} with settings \verb|--multi-trigger-cache| \verb|--full-saturate-quant|; the SMT solver \emph{z3-4.8.12}~\cite{deMouraBjorner:08} with its default settings;
the constrained horn clause (CHC) solver \emph{spacer}~\cite{KomuravelliGC14} with its default settings;
and the first-order theorem prover \emph{vampire-4.5.1}~\cite{RiazanovVoronkov02} with settings \verb|--memory_limit 8000| \verb|-p off|, i.e., with memory extended to 8GB and without proof output.
For the SMT/CHC solvers, we directly transformed the benchmarks into their respective formats. 
Vampire gets the same input as VLog transformed into the TPTP format~\cite{Sutcliffe17}.
Our experiments with vampire investigate how superposition reasoners perform on the hammered benchmarks compared to Datalog reasoners.

For the experiments, we used the TACAS~22 artifact evaluation VM (Ubuntu 20.04 with 8 GB RAM and a single processor core) on a system with an Intel Core i7-9700K CPU with eight 3.60GHz cores. 
Each tool got a time limit of 40 minutes for each problem.

The table in Fig.~\ref{fig:toolchainresults} lists for each benchmark problem: the name of the problem (Problem);
the type of conjecture (Q), i.e., whether the conjecture is existential $\exists$ or universal $\forall$; the status of the conjecture (Status);
number of clauses ($|N|$); maximum number of variables in a clause (vars); 
the size of the largest test-point set introduced by the sorted/original Hammer ($B^s$/$B^o$); the size of the hammered universal conjecture ($|\Delta_{\phi}|$/$|\Delta^o_{\phi}|$ for sorted/original); 
the remaining columns list the time needed by the tools to solve the benchmark problems. 
An entry "N/A" means that the benchmark example cannot be expressed in the tools input format, e.g., it is not possible to encode a universal conjecture (or, to be more precise, its negation) in the CHC format and SPASS-SPL-v0.6 is not sound when the problem contains integer variables. 
An entry "-" means that the tool ran out of time, ran out of memory, exited with an error or returned unknown.

The experiments show that SPASS-SPL (with the sorted Hammer) is orders of magnitudes faster than SPASS-SPL-v0.6 (with the original Hammer) on problems with universal conjectures.
On problems with existential conjectures, 
we cannot observe any major performance gain compared to the original Hammer.
Sometimes SPASS-SPL-v0.6 is even slightly faster (e.g. ecu\_e3). 
Potential explanations are:
First, the number of test points has a much larger impact on universal conjectures because the size of the hammered universal conjecture increases exponentially with the number of test points.
Second, our sorted Hammer needs to generate more abstracted theory facts than the original Hammer because the latter can reuse abstraction predicates for theory atoms that are identical upto variable renaming.
The sorted Hammer can reuse the same predicate only if variables also range over the same sets of test points, which we have not yet implemented.

Compared to the other tools, SPASS-SPL is the only one that solves all problems in reasonable time.
It is also the only solver that can decide in reasonable time whether a universal conjecture is \emph{not} a consequence.
This is not surprising because to our knowledge SPASS-SPL is the only theorem prover that implements a decision procedure for \HBS(\SB).
On the problems with existential conjectures, our tool-chain solves all of the problems in under a minute and with comparable times to the best tool for the problem.
The only exception are problems that contain a lot of superfluous clauses, i.e., clauses that are not needed to confirm/refute the conjecture.
The reason might be that VLog derives all facts for the input problem in a breadth-first way, which is not very efficient if there are a lot of superfluous clauses.
Vampire coupled with our sorted Hammer returns the best results for those problems.
Vampire performed best on the hammered problems among all first-order theorem provers we tested,
including iProver~\cite{Korovin08}, E~\cite{SchulzEtAl19}, and SPASS~\cite{WeidenbachEtAlSpass2009}. We tested all provers in default theorem
proving mode with adjusted memory limits.
The experiments with the first-order provers showed that our hammer also works reasonably well for them, but they do not scale well if the size and the complexity of the universal conjectures increases.
For problems with existential conjectures, the CHC solver spacer is often the best, but as a trade-off it is unable to handle universal conjectures.
The instantiation techniques employed by cvc4 are good for proving some universal conjectures, 
but both SMT solvers seem to be unable to disprove conjectures.


\section{Conclusion} \label{sec:conclusion}

We have presented an extension of our previous Datalog hammer~\cite{BrombergerEtAl21FROCOS} supporting
a more expressive input logic resulting in more elegant and more detailed supervisor formalizations, and
through a soft typing discipline supporting more efficient reasoning.
Our experiments show, compared to~\cite{BrombergerEtAl21FROCOS}, that our performance  on existential conjectures is at the same level as SMT
and CHC solvers. The complexity of queries we can handle in reasonable
time has significantly increased, see Section~\ref{sec:experiments}, Figure~\ref{fig:toolchainresults}.
Still SPASS-SPL is the only solver that can prove and disprove universal queries. 
The file interface between SPASS-SPL and VLog
has been replaced by a close coupling resulting in a more comfortable application. 

Our contribution here solves the third point for future work mentioned in~\cite{BrombergerEtAl21FROCOS}
although there is still room to also improve our soft typing discipline.
In the future, we want SPASS-SPL to produce explications that prove that its translations are correct.
Another direction is to exploit specialized Datalog expressions and techniques,
e.g., aggregation and stratified negation, to increase the efficiency of our tool-chain and to lift some restrictions from our input formulas.
Finally, our hammer can be seen as part of an overall reasoning methodology for the class of $\BS(\LA)$ formulas
which we presented in~\cite{BrombergerFW21}. We will implement and further develop this methodology and integrate
our Datalog hammer.


\smallskip\noindent
{\bf Acknowledgments:} This work was funded by DFG grant 389792660 as part of
\href{http://perspicuous-computing.science}{TRR~248 (CPEC)},
by BMBF in project \href{https://www.scads.de}{ScaDS.AI},
and by the \href{https://cfaed.tu-dresden.de/}{Center for Advancing Electronics Dresden} (cfaed).
We thank our anonymous reviewers for their constructive comments.

\bibliographystyle{splncs04}
\bibliography{paper}

\if 1\IsExtended
\section{Appendix}  \label{sec:appendix}

\subsection{Proofs and Auxiliary Lemmas}

\subsubsection{Auxiliary Lemma for the Proof of Lemma~\ref{lem:hbsfiniteashard}}

\begin{lemma}\label{lem:hbsfiniteashard}
Let $N$ be a set of $\HBS(\LA)$ clauses.
Let $Q^f$ be a predicate of arity one not occurring in $N$.
Let $y$ be a real variable not occurring in $N$.
Let $N' = \{(\Lambda \parallel \Delta \rightarrow H) \in N \mid H \neq \bot\} \cup \{(\Lambda \parallel \Delta \rightarrow Q^f(y)) \mid (\Lambda \parallel \Delta \rightarrow \bot) \in N\}$, i.e., the set of clauses $N$ just that we gave every clause $Q^f(y)$ as head literal that previously had no head literal.
Then $N$ is satisfiable if and only if argument position 1 of $Q^f$ is finite in $N'$.
\end{lemma}
\begin{proof}
Based on $\inta$, we can construct an interpretation $\inta'$ that is equivalent to $\inta$ except that it interprets $Q^f$ for all arguments as false and satisfies $N'$.
This is straightforward for all clauses $(\Lambda \parallel \Delta \rightarrow H) \in N'$ with $H \neq Q^f(y)$ because they also appear in $N$,
but it also holds for the clauses with $H = Q^f(y)$ because 
$\inta$ can only satisfy $(\Lambda \parallel \Delta \rightarrow \bot) \in N$ if $\bigwedge{A \in (\Lambda \cup \Delta)} A \sigma$ is interpreted as false by $\inta$.
Hence, $(\Lambda \parallel \Delta \rightarrow Q^f(y))$ is satisfied by $\inta'$, which means the set of derivable facts $\dfacts(Q^f,N)$ for $(Q^f,1)$ is empty and $(Q^f,1)$ is therefore finite.
Symmetrically, $Q^f$ is only finite if there exists at least one satisfiable interpretation for $N'$, 
where $\bigwedge{A \in (\Lambda \cup \Delta)} A \sigma$ is interpreted as false for every $(\Lambda \parallel \Delta \rightarrow Q^f(y)) \in N'$.
The reason is that any interpretation $\inta$ that satisfies $\bigwedge{A \in (\Lambda \cup \Delta)} A \sigma$ can derive all facts $Q^f(a)$ for $a \in \mathbb{R}$ (so infinitely many) from the clause $(\Lambda \parallel \Delta \rightarrow Q^f(y)) \in N'$.
However, if $\inta$ satisfies $N'$ and evaluates every $\bigwedge{A \in (\Lambda \cup \Delta)} A \sigma$ with $(\Lambda \parallel \Delta \rightarrow Q^f(y)) \in N'$ as false,
then $\inta$ also satisfies $N$.
This is straightforward for all clauses $(\Lambda \parallel \Delta \rightarrow H) \in N$ with $H \neq \bot$ because they also appear in $N'$,
but it also holds for the clauses with $H = \bot$ because 
$\inta$ evaluates $\bigwedge{A \in (\Lambda \cup \Delta)} A \sigma$ as false for those clauses.
Hence, $N$ is satisfiable.
\end{proof}

\subsubsection{Proof of Lemma~\ref{lem:hbsfiniteishard}}

\begin{proof}
Due to Lemma~\ref{lem:hbsfiniteishard}, we know that determining the finiteness of a predicate argument position can be as hard as determining the satisfiability of an $\HBS(\LA)$ clause set.
Thanks to~\cite{Downey1972,HorbachEtAl17ARXIV} we know that this is undecidable.
\end{proof}

\subsubsection{Converting Interpretations for $\HBS(\SB)\AG$ Problems}

\begin{lemma}\label{lem:interpretationtofinitelygrounded}
Every satisfying interpretation $\inta$ for $N$ is also a satisfying interpretation for $\aGnd(N)$.
\end{lemma}
\begin{proof}
We know that $\mGnd(\aGnd(N)) \subseteq \mGnd(N)$; at least after some canonical simplifications on the theory atoms. Therefore any interpretation that satisfies $N$ (and thus $\mGnd(N)$) also satisfies $\mGnd(\aGnd(N))$ and thus $\aGnd(N)$.
\end{proof}

\begin{lemma}\label{lem:interpretationfromfinitelygrounded}
Let $\inta$ be an interpretation satisfying the clause set $\aGnd(N)$. 
Then we can construct a satisfying interpretation $\inta'$ for $N$ such that $P^{\inta'} = \left\{\bar{a} \mid P(\bar{a}) \in \dfacts(P,\aGnd(N)) \right\}$.
\end{lemma}
\begin{proof}
Proof by contradiction. 
Suppose $\inta$ is an interpretation that satisfies $\aGnd(N)$ but not $\mGnd(N)$.
This would mean that there must exists a clause $(\Lambda \parallel \Delta \rightarrow H) \in N$ and a grounding $\sigma$ such that $\inta'$ does not satisfy $(\Lambda \parallel \Delta \rightarrow H) \sigma$, which would mean $\inta'$ satisfies $\Lambda \sigma$ and $\Delta \sigma$, but not $H \sigma$.
$\Delta \sigma$ is satisfied by $P^{\inta'}$ would imply that all atoms in $\Delta \sigma$ are derivable facts from $\aGnd(N)$.
However, since $\Lambda \sigma$ is satisfied and $\Delta \sigma$ consists of derivable facts from $\aGnd(N)$, $H \sigma$ should also be a derivable fact from $\aGnd(N)$.
This is a contradiction because this would imply that $H \sigma$ is actually satisfied.
\end{proof}

\subsubsection{Auxiliary Lemmas for the Proof of Theorem~\ref{thm:coveringdefinition}}

\begin{lemma}
\label{lem:uniforminequalityevaluation}
Let $\LAOP =  \{\leq,<,>,\geq, =, \neq\}$.
Let $(x \LAOP c) \in \connectedIneqs(Q,i,N)$ and let $a$ and $a'$ belong to $I \in \iPart(Q,i,N)$.
Then $a \LAOP c$ evaluates to true if and only if $a' \LAOP c$ evaluates to true.
\end{lemma}
\begin{proof}
We make a case distinction over the different cases for $\LAOP$:
\begin{itemize}
\item $\LAOP$ is $\leq$: this means that $\iBorders(Q,i,N)$ contains the interval borders $c]$ and $(c$.
Therefore, $I$ is either a subset of $(-\infty,c]$, i.e., all points in $I$ satisfy $(x \LAOP c)$ or $I \subseteq (c,\infty]$ so no points in $I$ satisfy $(x \LAOP c)$.
\item $\LAOP$ is $\geq$: this means that $\iBorders(Q,i,N)$ contains the interval borders $[c$ and $c)$.
Therefore, $I$ is either a subset of $[c,\infty)$, i.e., all points in $I$ satisfy $(x \LAOP c)$ or $I \subseteq (-\infty,c)$ so no points in $I$ satisfy $(x \LAOP c)$.
\item $\LAOP$ is $<$: this means that $\iBorders(Q,i,N)$ contains the interval borders $c)$ and $[c$.
Therefore, $I$ is either a subset of $(-\infty,c)$, i.e., all points in $I$ satisfy $(x \LAOP c)$ or $I \subseteq [c,\infty]$ so no points in $I$ satisfy $(x \LAOP c)$.
\item $\LAOP$ is $>$: this means that $\iBorders(Q,i,N)$ contains the interval borders $(c$ and $c]$.
Therefore, $I$ is either a subset of $(c,\infty)$, i.e., all points in $I$ satisfy $(x \LAOP c)$ or $I \subseteq (-\infty,c]$ so no points in $I$ satisfy $(x \LAOP c)$.
\item $\LAOP$ is $=$: this means that $\iBorders(Q,i,N)$ contains the interval $[c,c]$.
Therefore, $I$ is either $[c,c]$ or no point in $I$ satisfies $(x \LAOP c)$.
\item $\LAOP$ is $\neq$: this means that $\iBorders(Q,i,N)$ contains the interval $[c,c]$.
Therefore, $I$ is either $[c,c]$ and no point in $I$ satisfies $(x \LAOP c)$ or $I \neq [c,c]$ and all points in $I$ satisfy $(x \LAOP c)$.
\end{itemize}
\end{proof}

\begin{lemma}
\label{lem:uniformpredicateinterpretationZ}
Let $Q(\bar{a})$ be derivable from $N$ and let $a_i \in \mathbb{Z}$ belong to $I \in \iPart(Q,i,N)$.
Then $Q(\bar{a'})$ is also derivable from $N$, where $a'_j = a_j$ for $i \neq j$ and $a'_i \in I \cap  \mathbb{Z}$.
\end{lemma}
\begin{proof}
The case where $a'_i = a_i$ is trivial because $\bar{a'} = \bar{a}$.
We prove the case for $a'_i \neq a_i$ (and therefore also $I \neq [a_i,a_i]$) by structural induction over the derivations in $N$.
\begin{itemize}
\item The base case is that $Q(\bar{a})$ was only derived using one clause, i.e.,
      $N$ contains a clause $\Lambda \parallel \rightarrow Q(\bar{t})$ with 
      a grounding $\sigma$ such that $Q(\bar{a}) = Q(\bar{t}) \sigma$ and 
      $\Lambda \sigma$ evaluates to true.
      We can assume that $t_i \neq a_i$ because $I$ would be the interval $[a_i,a_i]$ otherwise.
      This means $t_i = x$ for a variable $x$.
      Based on Lemma~\ref{lem:uniforminequalityevaluation}, $\Lambda \sigma'$ 
      (with $x \sigma' = a'_i$ and $y \sigma' = y \sigma$ for all $y \neq x$) 
      must also evaluate to true.
      Therefore, $Q(\bar{a'})$ is also derivable from $N$.
\item The induction step is that $Q(\bar{a})$ was derived using a clause 
      $\Lambda \parallel \Delta \rightarrow Q(\bar{t}) \in N$ with a grounding $\sigma$ 
      such that $Q(\bar{a}) = Q(\bar{t}) \sigma$, 
      $\Lambda \sigma$ evaluates to true, and 
      all $P(\bar{s}) \sigma \in \Delta \sigma$ are derivable from $N$.      
      Moreover, we can assume that $t_i \neq a_i$ because $I$ would be the interval $[a_i,a_i]$ otherwise.
      Thus, $t_i = x$ for a variable $x$.
      This means we can again construct a substitution 
      $\sigma'$ with $x \sigma' = a'_i$ and $y \sigma' = y \sigma$ for all $y \neq x$.
      By induction we can assume that $P(\bar{s}) \sigma' \in \Delta \sigma'$ is derivable from $N$ because $P(\bar{s}) \sigma$ is derivable from $N$.    
      Due to Lemma~\ref{lem:uniforminequalityevaluation}, $\Lambda \sigma'$ is also satisfiable.
      Thus $Q(\bar{a'})$ is derivable from $N$ using 
      $(\Lambda \parallel \Delta \rightarrow Q(\bar{t})) \sigma'$.      
\end{itemize}
\end{proof}

\begin{lemma}
\label{lem:uniformpredicateinterpretationR}
Let $Q(\bar{a})$ be derivable from $N$ and let $a_i \not\in \mathbb{Z}$ belong to $I \in \iPart(Q,i,N)$.
Then $Q(\bar{a'})$ is also derivable from $N$, where $a'_j = a_j$ for $i \neq j$ and $a'_i \in I$.
\end{lemma}
\begin{proof}
The case where $a'_i = a_i$ is trivial because $\bar{a'} = \bar{a}$.
We prove the case for $a'_i \neq a_i$ (and therefore also $I \neq [a_i,a_i]$) by structural induction over the derivations in $N$.
\begin{itemize}
\item The base case is that $Q(\bar{a})$ was only derived using one clause, i.e.,
      $N$ contains a clause $\Lambda \parallel \rightarrow Q(\bar{t})$ with 
      a grounding $\sigma$ such that $Q(\bar{a}) = Q(\bar{t}) \sigma$ and 
      $\Lambda \sigma$ evaluates to true.
      We can assume that $t_i \neq a_i$ because $I$ would be the interval $[a_i,a_i]$ otherwise.
      This means $t_i = x$ for a variable $x$ and $\sort(x)=\RealS$ or we could not have derived a value $a_i \not\in \mathbb{Z}$.
      Based on Lemma~\ref{lem:uniforminequalityevaluation}, $\Lambda \sigma'$ 
      (with $x \sigma' = a'_i$ and $y \sigma' = y \sigma$ for all $y \neq x$) 
      must also evaluate to true.
      Therefore, $Q(\bar{a'})$ is also derivable from $N$.
\item The induction step is that $Q(\bar{a})$ was derived using a clause 
      $\Lambda \parallel \Delta \rightarrow Q(\bar{t}) \in N$ with a grounding $\sigma$ 
      such that $Q(\bar{a}) = Q(\bar{t}) \sigma$, 
      $\Lambda \sigma$ evaluates to true, and 
      all $P(\bar{s}) \sigma \in \Delta \sigma$ are derivable from $N$.      
      Moreover, we can assume that $t_i \neq a_i$ because $I$ would be the interval $[a_i,a_i]$ otherwise.
      Thus, $t_i = x$ for a variable $x$ and $\sort(x)=\RealS$ or we could not have derived a value $a_i \not\in \mathbb{Z}$.
      This means we can again construct a substitution 
      $\sigma'$ with $x \sigma' = a'_i$ and $y \sigma' = y \sigma$ for all $y \neq x$.
      By induction we can assume that $P(\bar{s}) \sigma' \in \Delta \sigma'$ is derivable from $N$ because $P(\bar{s}) \sigma$ is derivable from $N$.    
      Due to Lemma~\ref{lem:uniforminequalityevaluation}, $\Lambda \sigma'$ is also satisfiable.
      Thus $Q(\bar{a'})$ is derivable from $N$ using 
      $(\Lambda \parallel \Delta \rightarrow Q(\bar{t})) \sigma'$.      
\end{itemize}
\end{proof}

\begin{lemma}
\label{lem:welltypedgroundingrepresentation}
Let $\Lambda \parallel \Delta \rightarrow H$ be a clause in $N$.
Let $\sigma$ be a well-typed grounding over the most general test-point $\tpfunction$ $\beta^*$ such that $\Lambda \sigma$ evaluates to true and all atoms in $\Delta\sigma$ are derivable from $N$.
Then there exists a grounding $\sigma'$ that is a \emph{well-typed instance} over the $\tpfunction$ $\beta$, i.e., $\sigma \in \tg_{\beta}(\Lambda \parallel \Delta \rightarrow H)$, 
and from which we can extrapolate the interpretation for $(\Lambda \parallel \Delta \rightarrow H) \sigma$, i.e., $\bar{t}\sigma = \eta(P,\bar{t}\sigma')$ for all $P(\bar{t}) \in \atoms(\Delta \rightarrow H)$.
\end{lemma}
\begin{proof}
Before we start with the actual proof, let us repeat and clarify the definition of our extrapolation function and the definition of a well-typed instance over $\beta$.
The extrapolation $\eta(P,\bar{t}\sigma') = I'_{(P,1)} \times \ldots \times I'_{(P,n)}$ of a test-point vector $\bar{t}\sigma'$ is the cross product of the sets $I'_{(P,i)}$. 
One for each predicate argument position $(P,i)$.
If $(P,i)$ is finite, then $I'_{(P,i)} = t_i\sigma'$.
If $(P,i)$ is infinite and $t_i\sigma'$ an integer value, 
then $I'_{(P,i)} = \iPart(P,i,N) \cap \mathbb{Z}$.
If $(P,i)$ is infinite and $t_i\sigma'$ a non-integer value, 
then $I'_{(P,i)} = \iPart(P,i,N) \setminus \mathbb{Z}$.
A substitution $\sigma'$ for a clause $Y$ is a well-typed instance over $\beta$ if it guarantees for each variable $x$ that $x \sigma'$ is part of every test-point set (i.e., $x \sigma' \in \beta(P,i)$) of every argument position $(P,i)$ it occurs in (i.e., $(P,i)\in\depend(x,Y)$) and that $x \sigma' \in \sort(x)^\inta$. 
This means our proof only needs to show that any two argument positions $(P,i)$ and $(Q,j)$ in $Y$ that share the same variable $x$, share a test point $b \in \sort(x)^\inta$ such that $x \sigma$ can be extrapolated from $b$, i.e., for $b = x \sigma'$:
$b \in \beta(P,i)$, $b \in \beta(Q,j)$, $b \in \sort(x)^\inta$, 
and $x \sigma \in I'_{(P,i)}$ as well as $x \sigma \in I'_{(Q,j)}$.

First of all, $(x \sigma) \in \sort(x)^\inta$, $(x \sigma) \in \sort(P,i)^\inta$, and $(x \sigma) \in \sort(Q,j)^\inta$ because $\sigma$ is well-typed over $\beta^*$ and $N$ is well-typed.
Second of all, $H \sigma$ is derivable because the atoms in $\Delta\sigma$ are derivable from $N$ and $\Lambda \sigma$ evaluates to true.
Hence, $x \sigma$ must be a derivable value for both $(P,i)$ and $(Q,j)$.
This means our condition is trivial to satisfy if $(P,i)$ is finite and $(Q,j)$ is finite because 
$\dval(P,i,N) \subseteq \beta(P,i)$ and $\dval(Q,j,N) \subseteq \beta(Q,j)$.
So we can simply choose $b = (x \sigma)$ to satisfy our conditions.
The case where $(P,i)$ is finite and $(Q,j)$ is infinite works similarly.
Here we additionally need that $[x \sigma, x \sigma] \in \iPart(Q,i,N)$ because $\connectedArgs(P,i,N) = \connectedArgs(Q,j,N)$ and thus $(x \sigma) \in \beta(Q,j)$.
As a result, we can choose $b = (x \sigma)$ again to satisfy our conditions.
The reverse case, where $(P,i)$ is infinite and $(Q,j)$ is finite, works symmetrically.
In the case where $(P,i)$ and $(Q,j)$ are both infinite, 
we know that $\connectedArgs(P,i,N) = \connectedArgs(Q,j,N)$ because they share a variable in this clause and therefore $\testpoints(P,i,N) = \testpoints(Q,j,N)$.
Since we know that $(x \sigma) \in \sort(P,i)^\inta$ and $(x \sigma) \in \sort(Q,j)^\inta$, 
we know that $\beta(P,i,N) \cap \beta(Q,j,N)$ contains a (non-)integer value $b$ for the interval $I \in \iPart(P,i,N) = \iPart(Q,j,N)$ if $x \sigma \in I$ is a (non-)integer value.
Hence, we can choose $b = x \sigma'$ and satisfy our conditions.
\end{proof}

\begin{lemma}
\label{lem:testpointrepresentation}
Let $\bar{a}$ be a test-point vector for $Q$ over $\beta$, i.e., with $a_i \in \beta(Q,i)$ for all $i$.
Then $Q(\bar{a})$ is derivable from $\mGnd(N)$ if and only if $Q(\bar{a})$ is derivable from $\mGnd_{\beta}(N)$.
\end{lemma}
\begin{proof}
The first direction, $Q(\bar{a})$ is derivable from $\mGnd(N)$ if $Q(\bar{a})$ is derivable from $\mGnd_{\beta}(N)$, is straightforward because $\mGnd_{\beta}(N) \subseteq \mGnd(N)$.
So any derivation step in $\mGnd_{\beta}(N)$ can also be performed in $\mGnd(N)$.
The second direction follows by structural induction:
\item The base case is that $Q(\bar{a})$ was only derived using one clause, i.e.,
      $N$ contains a clause $\Lambda \parallel \rightarrow Q(\bar{t})$ with 
      a well-typed grounding $\sigma$ over the most general $\tpfunction$ $\beta^*$ 
      such that $Q(\bar{a}) = Q(\bar{t}) \sigma$ and $\Lambda \sigma$ evaluates to true.
      However, $(\Lambda \parallel \rightarrow Q(\bar{t}))\sigma$ is part of $\mGnd_{\beta}(N)$.
      Therefore, $Q(\bar{a})$ is also derivable from $\mGnd_{\beta}(N)$.
\item The induction step is that $Q(\bar{a})$ was derived using a clause 
      $\Lambda \parallel \Delta \rightarrow Q(\bar{t}) \in N$ with a well-typed grounding $\sigma$ 
      over the most general $\tpfunction$ $\beta^*$ 
      such that $Q(\bar{a}) = Q(\bar{t}) \sigma$, 
      $\Lambda \sigma$ evaluates to true, and 
      all $P(\bar{s}) \sigma \in \Delta \sigma$ are derivable from $N$.
      Since all $P(\bar{s}) \sigma \in \Delta \sigma$ are derivable and 
      $\Lambda \sigma$ evaluates to true, 
      we can use Lemma~\ref{lem:welltypedgroundingrepresentation} to construct a grounding $\sigma'$ 
      that is well-typed over $\beta$ and 
      such that $\bar{s} \sigma = \eta(P,\bar{s} \sigma')$ for all $P(\bar{s}) \sigma \in \Delta \sigma$.
      By Lemmas~\ref{lem:uniformpredicateinterpretationZ} 
      and~\ref{lem:uniformpredicateinterpretationR}, 
      we can assume that $P(\bar{s}) \sigma'$ is derivable from $N$ because 
      $P(\bar{s}) \sigma$ is derivable from $N$.      
      Following that we can assume by induction that $P(\bar{s}) \sigma' \in \Delta \sigma'$ is derivable from $\mGnd_{\beta}(N)$ because $P(\bar{s}) \sigma'$ is derivable from $N$.    
      Due to Lemma~\ref{lem:uniforminequalityevaluation}, $\Lambda \sigma'$ is also satisfiable.
      Thus $Q(\bar{a}) = Q(\bar{t}) \sigma' = Q(\bar{t}) \sigma$ is derivable 
      from $\mGnd_{\beta}(N)$ using $(\Lambda \parallel \Delta \rightarrow Q(\bar{t})) \sigma'$.    
\end{proof}

\begin{lemma}
\label{lem:satgrounding}
The $\tpfunction$ $\beta$ covers $N$, i.e., $\mGnd_{\beta}(N)$ is equisatisfiable to $N$.
\end{lemma}
\begin{proof}
If $N$ is satisfiable, then $\mGnd(N)$ is satisfiable.
Hence $\mGnd_{\beta}(N)$ is also satisfiable because $\mGnd_{\beta}(N)\subseteq \mGnd(N)$ .
For the reverse direction we assume that $\mGnd_{\beta}(N)$ is satisfiable.
Then we show that we can extrapolate a new interpretation $\inta$ from the derivable facts of $\mGnd_{\beta}(N)$ so it satisfies $\mGnd(N)$ and thus $N$.
The extrapolation is defined as follows:
$P^{\inta} = \{\bar{b} \mid \bar{a} \in \dfacts(P,\mGnd_{\beta}(N)) \text{ and } \bar{b} \in \eta(P,\bar{a})\}$.
The interpretation $\inta$ satisfies every clause $(\Lambda \parallel \Delta \rightarrow H) \in \mGnd(N)$, i.e., every well-typed grounding $(\Lambda \parallel \Delta \rightarrow H) \sigma$ of every clause $(\Lambda \parallel \Delta \rightarrow H) \in N$ over the most general $\tpfunction$ $\beta^*$, 
due to one of three reasons:
(i)~$\Lambda \sigma$ evaluates to false. 
(ii)~An atom $A$ from $\Delta \sigma$ is not derivable from $N$ and since $\mGnd_{\beta}(N)$ is a subset of $\mGnd(N)$ the atom $A$ is also not derivable from $\mGnd_{\beta}(N)$. By definition of $\inta$ this means $\Delta \sigma$ is interpreted as false and the clause is thus satisfied.
(iii)~We can assume that $\Lambda \sigma$ evaluates to true and that all atoms $A$ from $\Delta \sigma$ are derivable from $N$.
      This means we can use Lemma~\ref{lem:welltypedgroundingrepresentation} to construct a 
      grounding $\sigma'$ that is well-typed over $\beta$ and 
      such that $\bar{s} \sigma = \eta(P,\bar{s} \sigma')$ 
      for all $P(\bar{s}) \sigma \in \atoms((\Delta \rightarrow H) \sigma)$.
      Due to Lemma~\ref{lem:uniforminequalityevaluation}, 
      this means $\Lambda \sigma'$ evaluates to true.
      Together with Lemmas~\ref{lem:uniformpredicateinterpretationZ} 
      and~\ref{lem:uniformpredicateinterpretationR} 
      this implies that all atoms in $\Delta \sigma'$ and $H \sigma'$ are derivable from $\mGnd_{\beta}(N)$.
      Hence, $\inta$ satisfies $H \sigma = P(\bar{s})\sigma$ and the full clause because 
      $\bar{s} \sigma = \eta(P,\bar{s} \sigma')$ and 
      $H \sigma'$ is derivable from $\mGnd_{\beta}(N)$.
\end{proof}

\begin{lemma}
\label{lem:existgrounding}
The $\tpfunction$ $\beta$ covers an existential conjecture $N \models \exists \bar{x}. Q(\bar{x})$, i.e., 
$\mGnd_{\beta}(N) \cup \{\mGnd_{\beta}(\parallel Q(\bar{x}) \rightarrow \bot)\}$ is satisfiable if and only if $N \models \exists \bar{x}. Q(\bar{x})$ is false.
\end{lemma}
\begin{proof}
$N \models \exists \bar{x}. Q(\bar{x})$ if $N \cup \{Q(\bar{x}) \rightarrow \bot\}$ is unsatisfiable. 
Hence, Lemma~\ref{lem:satgrounding} shows that $\beta$ covers the existential conjecture.
\end{proof}

\begin{lemma}
\label{lem:univgrounding}
The $\tpfunction$ $\beta$ covers a universal conjecture $N \models \forall \bar{x}. Q(\bar{x})$, i.e.,
$\mGnd_{\beta}(N) \cup N_C$ is satisfiable if and only if $N \models \forall \bar{x}. Q(\bar{x})$ is false.
Here $N_C$ is the set $\{\parallel \mGnd_{\beta}(Q(\bar{x})) \rightarrow \bot \}$ if $\eta$ is complete for $Q$ or the empty set otherwise.
\end{lemma}
\begin{proof}
We split the proof into four parts:
\begin{enumerate}
\item Assume $\mGnd_{\beta}(N) \cup N_C$ is satisfiable and $\eta$ is complete for $Q$.
      This means $\mGnd_{\beta}(N)$ alone is also satisfiable and by 
      Lemma~\ref{lem:satgrounding} we know that $N$ is, too.
      Since $N_C$ is also satisfiable, 
      there must exist an instance $Q(\bar{a}) \in \mGnd_{\beta}(Q(\bar{x}))$ 
      that is not derivable from $\mGnd_{\beta}(N)$.
      By Lemma~\ref{lem:testpointrepresentation}, 
      we know this means that $Q(\bar{a})$ is also not derivable for $N$.
      Hence $N \models \forall \bar{x}. Q(\bar{x})$ is false.
\item Assume $\mGnd_{\beta}(N) \cup N_C$ is satisfiable and $\eta$ does not cover $Q$.
      This means that our over-approximation detected that we cannot derive all instances for $Q$.
      The only thing left to show is that $\mGnd_{\beta}(N)$ is satisfiable implies that
      $N$ is satisfiable.
      This follows from Lemma~\ref{lem:satgrounding}.
\item Assume $\mGnd_{\beta}(N)$ is unsatisfiable.
      Then due to Lemma~\ref{lem:satgrounding}, $N$ is also unsatisfiable and hence the universal conjecture holds.
\item Assume $\mGnd_{\beta}(N)$ is satisfiable, but $\mGnd_{\beta}(N) \cup N_C$ is not.
      First of all, this means $\eta$ is complete for $Q$ and 
      $N$ is satisfiable (Lemma~\ref{lem:satgrounding}).
      Moreover, it means that all facts in $\mGnd_{\beta}(Q(\bar{x}))$ are derivable from $\mGnd_{\beta}(N)$.
      Since $\eta$ is complete for $Q$ and all instances of $Q$ over our test-points are derivable,
      Lemmas~\ref{lem:uniformpredicateinterpretationZ} and~\ref{lem:uniformpredicateinterpretationR} 
      imply that all groundings of $Q$ (i.e., all facts in $\mGnd(Q(\bar{x}))$) 
      are derivable from $N$.
      Hence, the universal conjecture holds.
\end{enumerate}
\end{proof}

\subsubsection{Proof of Theorem~\ref{thm:coveringdefinition}}

\begin{proof}
See Lemmas~\ref{lem:satgrounding},~\ref{lem:existgrounding}, and~\ref{lem:univgrounding}.
\end{proof}

\subsubsection{Auxiliary Lemmas for the Proof of Lemma~\ref{lem:hammeringequivalence}}

\begin{lemma}
\label{lem:hammerequiderivability}
Let $Q$ be a predicate in $N$.
Let $\beta$ be a finite and covering $\tpfunction$ for $N$.
Let $N_H := \tren_N(N) \cup \tfacts(N,\beta) \cup \sfacts(N,\beta)$ be the hammered version of $N$.
Then any fact $Q(\bar{a})$ derivable from $N_H$ is also derivable from $\mGnd_{\beta}(N)$ and vice versa.
\end{lemma}
\begin{proof}
First, we prove that $Q(\bar{a})$ is derivable from $N_H$ if $Q(\bar{a})$ is derivable from $\mGnd_{\beta}(N)$.
We prove this by structural induction over the derivations in $\mGnd_{\beta}(N)$.
\begin{itemize}
\item The base case is that $Q(\bar{a})$ was only derived using one clause, i.e.,
      $N$ contains a clause $\Lambda \parallel \rightarrow Q(\bar{t})$ that has
      a well-typed grounding $\sigma$ over $\beta$ such that $Q(\bar{a}) = (Q(\bar{t}) \sigma)$ 
      and $\Lambda \sigma$ evaluates to true.
      However, this also means that $\tren_N(N)$ contains a clause $\Delta_T, \Delta_S \rightarrow Q(\bar{t})$.
      $\Delta_T$ are all abstracted theory literals and all their groundings
      $\Delta_T \sigma$ appear as facts in $\tfacts(N,\beta)$ since $\Lambda \sigma$ evaluates to true.
      $\Delta_S$ are all sort literals $Q_{(Q,i,S)}(x)$ such that $t_i = x$ and 
      $\sfacts(N,\beta)$ contains  $Q_{(Q,i,S)}(x \sigma)$ because it must hold that $x \sigma \in \beta(Q,i) \cap S^\inta$ or $\sigma$ would not be well-typed over $\beta$.
      Hence, all $A \in (\Delta_T \sigma \cup \Delta_S \sigma)$ are derivable from $N_H$
      and therefore $Q(\bar{a})$ is derivable via $(\Delta_T, \Delta_S \rightarrow Q(\bar{t}))\sigma$ from $N_H$.
\item The induction step is that $Q(\bar{a})$ was derived using a clause 
      $(\Lambda \parallel \Delta \rightarrow Q(\bar{t})) \in N$ with 
      a well-typed grounding $\sigma$ over $\beta$ such that $Q(\bar{a}) = Q(\bar{t}) \sigma$, 
      $\Lambda \sigma$ evaluates to true, and 
      all $P(\bar{s}) \sigma \in \Delta \sigma$ are derivable from $\mGnd_{\beta}(N)$ and $N_H$.
      This also means that $\tren_N(N)$ contains a clause $\Delta_T, \Delta_S, \Delta \rightarrow Q(\bar{t})$.
      $\Delta_T$ are all abstracted theory literals and all their groundings
      $\Delta_T \sigma$ appear as facts in $\tfacts(N,\beta)$ since $\Lambda \sigma$ evaluates to true.
      $\Delta_S$ are all sort literals $Q_{(Q,i,S)}(x)$ such that $t_i = x$ and 
      $\sfacts(N,\beta)$ contains  $Q_{(Q,i,S)}(x \sigma)$ because it must hold that $x \sigma \in \beta(Q,i) \cap S^\inta$ or $\sigma$ would not be well-typed over $\beta$.
      Hence, all $A \in (\Delta_T \sigma \cup \Delta_S \sigma \cup \Delta \sigma)$ 
      are derivable from $N_H$ and therefore $Q(\bar{a})$ is derivable via 
      $(\Delta_T, \Delta_S, \Delta \rightarrow Q(\bar{t}))\sigma$ from $N_H$.     
\end{itemize}
Second, we prove that $Q(\bar{a})$ is derivable from $\mGnd_{\beta}(N)$ if $Q(\bar{a})$ is derivable from $N_H$.
We prove this again by structural induction over the derivations in $N_H$.
\begin{itemize}
\item The base case is that $Q(\bar{a})$ was only derived using one clause in $\tren_N(N)$, i.e.,
      $\tren_N(N)$ contains a clause $\Delta_T, \Delta_S \rightarrow Q(\bar{t})$, 
      there exists a grounding $\sigma$ such that $Q(\bar{a}) = (Q(\bar{t}) \sigma)$ 
      and such that $\Delta_T \sigma \subseteq \tfacts(N,\beta)$ and $\Delta_S \sigma \subseteq \sfacts(N,\beta)$.
      However, this also means that $N$ contains a clause $\Lambda \parallel \rightarrow Q(\bar{t})$
      such that $(\Delta_T, \Delta_S \rightarrow Q(\bar{t})) = \tren_N(\Lambda \parallel \rightarrow Q(\bar{t}))$.
      Therefore, $\Delta_T \sigma \subseteq \tfacts(N,\beta)$ implies that $\Lambda \sigma$
      evaluates to true and that $\sigma$ is well-typed over $\beta$ for the variables $x \in \vars(\Lambda)$.
      All other variables $x \in \vars(Q(\bar{t})) \setminus \vars(\Lambda)$ (with $t_i = x$)
      were added by $\tren_N$ through sort literals $Q_{(Q,i,S)}(x)$ to $\Delta_S$ and 
      $\Delta_S \sigma \subseteq \sfacts(N,\beta)$ ensures that $\sigma$ is also well-typed 
      over $\beta$ for those variables.
      Hence, $\sigma$ is a well-typed grounding over $\beta$ for 
      $\Lambda \parallel \rightarrow Q(\bar{t})$ and $\Lambda \sigma$ evaluates to true.
      Thus, $Q(\bar{a})$ is derivable via $(\Lambda \parallel \rightarrow Q(\bar{t}))\sigma$ from $\mGnd_{\beta}(N)$.
\item The induction step is that $Q(\bar{a})$ was derived using a clause 
      $(\Delta_T, \Delta_S, \Delta \rightarrow Q(\bar{t})) \in \tren_N(N)$,
      all $P(\bar{s}) \sigma \in \Delta \sigma$ are derivable from $\mGnd_{\beta}(N)$ and $N_H$, 
      and there exists a grounding $\sigma$ such that $Q(\bar{a}) = (Q(\bar{t}) \sigma)$ 
      and such that $\Delta_T \sigma \subseteq \tfacts(N,\beta)$ and $\Delta_S \sigma \subseteq \sfacts(N,\beta)$.
      This also means that $N$ contains a clause $\Lambda \parallel \Delta \rightarrow Q(\bar{t})$
      such that $(\Delta_T, \Delta_S, \Delta \rightarrow Q(\bar{t})) = \tren_N(\Lambda \parallel \Delta \rightarrow Q(\bar{t}))$.
      Therefore, $\Delta_T \sigma \subseteq \tfacts(N,\beta)$ implies that $\Lambda \sigma$
      evaluates to true and that $\sigma$ is well-typed over $\beta$ for the variables $x \in \vars(\Lambda)$.
      The grounding $\sigma$ is also well-typed over $\beta$ for all variables $x \in \vars(Q(\bar{t})) \setminus \vars(\Lambda)$ with $\sort(x) = \RealS$ that appear in at least one literal $P(\bar{s}) \in \Delta$ with $s_j = x$ or otherwise the groundings $R(\bar{s'}) \sigma$ of all literals $R(\bar{s'})$ containing $x$ would not be derivable from $\mGnd_{\beta}(N)$.
      The grounding $\sigma$ is also well-typed over $\beta$ for all variables $x \in \vars(Q(\bar{t})) \setminus \vars(\Lambda)$ with $\sort(x) = \IntS$ that appear in a literal $P(\bar{s}) \in \Delta$ with $s_j = x$ and $\sort(P,j) = \IntS$ or otherwise the groundings $R(\bar{s'}) \sigma$ of all literals $R(\bar{s'})$ containing $x$ would not be derivable from $\mGnd_{\beta}(N)$.
      All other variables $x \in \vars(Q(\bar{t})) \setminus \vars(\Lambda)$ (with $t_i = x$)
      were added by $\tren_N$ through sort literals $Q_{(Q,i,S)}(x)$ to $\Delta_S$ and 
      $\Delta_S \sigma \subseteq \sfacts(N,\beta)$ ensures that $\sigma$ is also well-typed 
      over $\beta$ for those variables.
      Hence, $\sigma$ is a well-typed grounding over $\beta$ for 
      $\Lambda \parallel \Delta \rightarrow Q(\bar{t})$, $\Lambda \sigma$ evaluates to true, 
      and all atoms in $\Delta \sigma$ are derivable from $\mGnd_{\beta}(N)$.
      Thus, $Q(\bar{a})$ is derivable via $(\Lambda \parallel \Delta \rightarrow Q(\bar{t}))\sigma$ 
      from $\mGnd_{\beta}(N)$.    
\end{itemize}
\end{proof}

\begin{lemma}\label{lem:hammeringnoconjequivalence}
$N$ is equisatisfiable to its hammered version $\tren_N(N) \cup \tfacts(N,\beta) \cup \sfacts(N,\beta)$.
\end{lemma}
\begin{proof}
Let $N_H := \tren_N(N) \cup \tfacts(N,\beta) \cup \sfacts(N,\beta)$.
We split the proof into two parts:
\begin{enumerate}
\item Assume $\mGnd_{\beta}(N)$ is satisfiable. 
Then the interpretation $\inta$ defined by the following rules satisfies $N_H$:
$P^{\inta} := \left\{\bar{a} \mid P(\bar{a}) \in \dfacts(P,\mGnd_{\beta}(N)) \right\}$ if $P$ is a predicate that appears in $N$,
$P^{\inta} := \left\{\bar{a} \mid P(\bar{a}) \in \tfacts(N,\beta) \right\}$ if $P$ is a predicate that appears in $\tfacts(N,\beta)$, and 
$P^{\inta} := \left\{\bar{a} \mid P(\bar{a}) \in \sfacts(N,\beta) \right\}$ if $P$ is a predicate that appears in $\sfacts(N,\beta)$.
We prove this by contradiction: 
Suppose $\inta$ is an interpretation that satisfies $\mGnd_{\beta}(N)$ but not $N_H$.
Naturally, $\inta$ satisfies any of the facts in $\tfacts(N,\beta)$ and in $\sfacts(N,\beta)$.
This means that there must exist a clause $(\Lambda \parallel \Delta \rightarrow H) \in N$, a 
clause $(\Delta_T, \Delta_S, \Delta \rightarrow H) \in \tren_N(N)$ with $(\Delta_T, \Delta_S, \Delta \rightarrow H) = \tren_N(\Lambda \parallel \Delta \rightarrow H)$ and a grounding $\sigma$ 
such that $\inta$ does not satisfy $(\Delta_T, \Delta_S, \Delta \rightarrow H) \sigma$.
Formally, $\inta$ does not satisfy $(\Delta_T, \Delta_S, \Delta \rightarrow H) \sigma$ means $\inta$ satisfies $\Delta_T \sigma$, 
$\inta$ satisfies $\Delta_S \sigma$, $\inta$ satisfies $\Delta \sigma$, but $\inta$ does not satisfy $H \sigma$.
Since $\inta$ satisfies $\Delta_T \sigma$, we know by definition of $\inta$ and $\tfacts(N,\beta)$ that all atoms in $\Delta_T \sigma$ must be derivable from $\tfacts(N,\beta)$.
Since $\inta$ satisfies $\Delta_S \sigma$, we know by definition of $\inta$ and $\sfacts(N,\beta)$ that all atoms in $\Delta_S \sigma$ must be derivable from $\sfacts(N,\beta)$.
Since $\inta$ satisfies $\Delta \sigma$, we know by definition of $\inta$ that all facts $\Delta \sigma$ can be derived from $\mGnd_{\beta}(N)$.
By Lemma~\ref{lem:hammerequiderivability}, this means that all facts $\Delta \sigma$ can be derived from $N_H$, too.
Thus, $H$ must be derivable from $N_H$ and, by Lemma~\ref{lem:hammerequiderivability}, $H$ must be derivable by $\mGnd_{\beta}(N)$.
Hence, $H$ would need to be interpreted as true by $\inta$, which is the contradiction we were looking for.
Thus, $\inta$ satisfies $N_H$.
\item Assume $N_H$ is satisfiable.
Then the interpretation $\inta$ defined by the following rules satisfies $\mGnd_{\beta}(N)$:
$P^{\inta} := \left\{\bar{a} \mid P(\bar{a}) \in \dfacts(P,N_H) \right\}$ if $P$ is a predicate that appears in $N$,
$P^{\inta} := \left\{\bar{a} \mid P(\bar{a}) \in \tfacts(N,\beta) \right\}$ if $P$ is a predicate that appears in $\tfacts(N,\beta)$, and 
$P^{\inta} := \left\{\bar{a} \mid P(\bar{a}) \in \sfacts(N,\beta) \right\}$ if $P$ is a predicate that appears in $\sfacts(N,\beta)$.
We prove this by contradiction: 
Suppose $\inta$ is an interpretation that satisfies $N_H$ but not $\mGnd_{\beta}(N)$.
This means that there must exist a clause $(\Lambda \parallel \Delta \rightarrow H) \in N$, a 
clause $(\Delta_T, \Delta_S, \Delta \rightarrow H) \in \tren_N(N)$ with $(\Delta_T, \Delta_S, \Delta \rightarrow H) = \tren_N(\Lambda \parallel \Delta \rightarrow H)$ and a well-typed grounding $\sigma$ over $\beta$ such that $\inta$ does not satisfy $(\Lambda \parallel \Delta \rightarrow H) \sigma$. 
Formally, $\inta$ does not satisfy $(\Lambda \parallel \Delta \rightarrow H) \sigma$ means $\Lambda \sigma$ evaluates to true, 
$\inta$ satisfies $\Delta \sigma$, but not $H \sigma$.
Since $\inta$ satisfies $\Delta \sigma$, we know by definition of $\inta$ that all facts $\Delta \sigma$ can be derived from $N_H$.
By Lemma~\ref{lem:hammerequiderivability}, this means that all facts $\Delta \sigma$ can be derived from $\mGnd_{\beta}(N)$, too.
Thus, $H$ must be derivable from $\mGnd_{\beta}(N)$ and, by Lemma~\ref{lem:hammerequiderivability}, $H$ must be derivable by $N_H$.
Hence, $H$ would need to be interpreted as true by $\inta$, which is the contradiction we were looking for.
Thus, $\inta$ satisfies $N_H$.
\end{enumerate}
\end{proof}

\begin{lemma}\label{lem:hammeringexisequivalence} 
The conjecture $N \models \exists \bar{y}. Q(\bar{y})$ is false iff $N_D = \tren_N'(N') \cup \tfacts(N',\beta) \cup \sfacts(N',\beta)$ is satisfiable with $N' = N \cup \{Q(\bar{y}) \rightarrow \bot\}$. 
\end{lemma}
\begin{proof}
$N \models \exists \bar{x}. Q(\bar{x})$ if $N \cup \{Q(\bar{x}) \rightarrow \bot\}$ is unsatisfiable. 
Hence, Lemma~\ref{lem:hammeringnoconjequivalence} shows that $N \models \exists \bar{y}. Q(\bar{y})$ is false iff $N_D = \tren_N'(N') \cup \tfacts(N',\beta) \cup \sfacts(N',\beta)$ is satisfiable.
\end{proof}

\begin{lemma}\label{lem:hammeringunivequivalence} 
The conjecture $N \models \forall \bar{y}. Q(\bar{y})$ is false iff $N_D = \tren_N(N) \cup \tfacts(N,\beta) \cup \sfacts(N,\beta) \cup N_C$ is satisfiable. 
\end{lemma}
\begin{proof}
Let $N_H := \tren_N(N) \cup \tfacts(N,\beta) \cup \sfacts(N,\beta)$.
We will prove that $\psi = \mGnd_{\beta}(N) \cup N_C$ is equisatisfiable to $N_H \cup N_C$. Then we get from Theorem~\ref{thm:coveringdefinition} that $N \models \forall \bar{y}. Q(\bar{y})$ is false if and only if the hammered version $N_H \cup N_C$ is satisfiable. (The case for $N$ without conjecture follows from Lemma~\ref{lem:hammeringnoconjequivalence}.)\newline
We split the proof into four parts:
\begin{enumerate}
\item Assume $N_H \cup N_C$ is satisfiable and $\eta$ does not cover $Q$.
      This means that our over-approximation detected that we cannot derive all instances for $Q$.
      The only thing left to show is that $N_H$ is satisfiable implies that
      $\mGnd_{\beta}(N)$ is satisfiable.
      This follows from Lemma~\ref{lem:hammeringnoconjequivalence}.
\item Assume $N_H \cup N_C$ is satisfiable and $\eta$ is complete for $Q$.
      This means $N_H$ alone is also satisfiable and by 
      Lemma~\ref{lem:hammeringnoconjequivalence} this means that $\mGnd_{\beta}(N)$ is, too.
      Since $N_C$ is also satisfiable, 
      there must exist an instance $Q(\bar{a}) \in \mGnd_{\beta}(Q(\bar{x}))$ 
      that is not derivable from $N_H$.
      By Lemma~\ref{lem:hammerequiderivability}, 
      we know this means that $Q(\bar{a})$ is also not derivable for $\mGnd_{\beta}(N)$.
\item Assume $N_H$ is unsatisfiable.
      Then due to Lemma~\ref{lem:hammeringnoconjequivalence}, $\mGnd_{\beta}(N)$ is also unsatisfiable and hence $N_H \cup N_C$ and $\mGnd_{\beta}(N) \cup N_C$ are both unsatisfiable.
\item Assume $N_H $ is satisfiable, but $N_H \cup N_C$ is not.
      This means $\eta$ is complete for $Q$ and by Lemma~\ref{lem:hammeringnoconjequivalence} this means that $\mGnd_{\beta}(N)$ is satisfiable, too.      
      Moreover, it means that all facts in $\mGnd_{\beta}(Q(\bar{x}))$ are derivable from $N_H$.
      By Lemma~\ref{lem:hammerequiderivability}, this implies that all facts in $\mGnd_{\beta}(Q(\bar{x}))$ are also derivable from $\mGnd_{\beta}(N)$.
      Hence, $\mGnd_{\beta}(N) \cup N_C$ is also not satisfiable.
\end{enumerate}
\end{proof}

\subsubsection{Proof of Lemma~\ref{lem:hammeringequivalence}}

\begin{proof}
See Lemmas~\ref{lem:hammeringnoconjequivalence}, \ref{lem:hammeringexisequivalence}, and \ref{lem:hammeringunivequivalence}.
\end{proof}

$\connectedArgs(P,i,N)$ the \emph{set of connected argument positions} and 
by $\connectedIneqs(P,i,N)$

\subsection{Pseudo-Code Algorithms for the steps of the Sorted Datalog Hammer}

\subsubsection{Checking whether a set of clauses is reducible to $\HBS(\SB)$}

\begin{center}
\begin{tabular}{l}
  $\text{IsReducible}(N,\aval)$\\
  for all Horn clauses $\Lambda \parallel \Delta \rightarrow H \in N$\\
  $\quad$ for all variables $x$ appearing in $\Lambda \parallel \Delta \rightarrow H$\\
  $\quad$ $\quad$ $\text{is\_finite(x)} := \texttt{false}$;\\
  $\quad$ for all atoms $P(t_1,\ldots,t_n) \in \Delta$ and $1 \leq i \leq n$\\  
  $\quad$ $\quad$ if [$t_i$ is a variable and $\aval(P,i,N)$ is finite] then\\
  $\quad$ $\quad$ $\quad$ $\text{is\_finite(x)} := \texttt{true}$;\\   
  $\quad$ for all inequalities $\lambda$ in $\Lambda$\\
  $\quad$ $\quad$ $\text{inftvars} = 0$;\\
  $\quad$ $\quad$ for all variables $x$ appearing in $\lambda$\\ 
  $\quad$ $\quad$ $\quad$ if [$\text{is\_finite(x)} = \texttt{false}$] then\\
  $\quad$ $\quad$ $\quad$ $\quad$ $\text{inftvars}++$;\\
  $\quad$ $\quad$ if [$\text{inftvars} > 1$] then\\
  $\quad$ $\quad$ $\quad$ return \texttt{false};\\
  return \texttt{true};
\end{tabular}
\end{center}

\subsubsection{Computing $\connectedArgs(P,i,N)$ and $\connectedIneqs(P,i,N)$}

\begin{center}
\begin{tabular}{l}
  $\text{FindConnections}(N)$\\
  for all predicates $P$ and argument positions $i$ for $P$\\
  $\quad$ $\connectedArgs(P,i,N) := \{(P,i)\}$; $\connectedIneqs(P,i,N) := \emptyset$\\
  for all Horn clauses $\Lambda \parallel \Delta \rightarrow H \in N$\\
  $\quad$ for all atoms $P(t_1, \ldots, t_n)$, $Q(s_1, \ldots, s_m)$ in $\Delta \rightarrow H$\\
  $\quad$ $\quad$ for all positions $1 \leq i \leq n$ and $1 \leq j \leq m$\\
  $\quad$ $\quad$ $\quad$ if [$t_i = s_j$ and $t_i $is a variable] then\\
  $\quad$ $\quad$ $\quad$ $\quad$ $\connectedArgs(P,i,N) := \connectedArgs(P,i,N) \cup \connectedArgs(Q,j,N)$;\\
  $\quad$ $\quad$ $\quad$ $\quad$ for all $(R,k) \in \connectedArgs(P,i,N)$\\
  $\quad$ $\quad$ $\quad$ $\quad$  $\quad$ $\connectedArgs(R,k,N) := \connectedArgs(P,i,N)$;\\
  for all Horn clauses $\Lambda \parallel \Delta \rightarrow H \in N$\\
  $\quad$ for all atoms $P(t_1, \ldots, t_n)$ in $\Delta \rightarrow H$\\
  $\quad$ $\quad$ if [$(t_i = c)$] then\\
  $\quad$ $\quad$ $\quad$ $\connectedIneqs(P,i,N) := \connectedIneqs(P,i,N) \cup \{(x = c)\}$;\\
  $\quad$ $\quad$ else if [$t_i$ is a variable an $(P,i)$ is infinite] then\\
  $\quad$ $\quad$ $\quad$ for all inequalities $\lambda$ in $\Lambda$ that contain $t_i$\\
  $\quad$ $\quad$ $\quad$ $\quad$ let $x_1, \ldots, x_n$ be the other variables appearing in $\lambda$; \\ 
  $\quad$ $\quad$ $\quad$ $\quad$ let $(Q_{j,1},i_{j,1}), \ldots, (Q_{j,m},i_{j,m})$ be connected to $x_j$ in $\Delta$;\\
  $\quad$ $\quad$ $\quad$ $\quad$ let $S_j := \bigcap_k \aval(Q_{j,k},i_{j,k},N)$;\\
  $\quad$ $\quad$ $\quad$ $\quad$ for all $\sigma$ with $\sigma(x_j) \in S_j$ for $1 \leq j \leq n$\\
  $\quad$ $\quad$ $\quad$ $\quad$ $\quad$ $\lambda'$ is $\lambda \sigma$ simplified to a variable bound;\\
  $\quad$ $\quad$ $\quad$ $\quad$ $\quad$ $\connectedIneqs(P,i,N) := \connectedIneqs(P,i,N) \cup \{\lambda'\}$;\\
  $\quad$ $\quad$ for all $(R,k) \in \connectedArgs(P,i,N)$\\
  $\quad$ $\quad$ $\quad$ $\connectedIneqs(R,k,N) := \connectedIneqs(P,i,N)$;\\
\end{tabular}
\end{center}

\subsubsection{Picking test points from the interval partitions}

\begin{center}
\begin{tabular}{l}
  $\text{PickTestpoints}(N)$\\
  for all predicates $P$ and argument positions $i$ for $P$\\
  $\quad$ $\testpoints(P,i,N) := \emptyset$;\\
  $\quad$ if [$|\aval(P,i,N)| \in \mathbb{N}$] then\\
  $\quad$ $\quad$ $\testpoints(P,i,N) := \aval(P,i,N)$;\\
  for all predicates $P$ and argument positions $i$ for $P$\\
  $\quad$ if [$|\aval(P,i,N)| = \infty$ and $\testpoints(P,i,N) = \emptyset$] then \\
  $\quad$ $\quad$ for all $I \in \iPart(P,i,N)$;\\
  $\quad$ $\quad$ $\quad$ if [$I \cap \Int \neq \emptyset$] then \\
  $\quad$ $\quad$ $\quad$ $\quad$ pick $a \in I \cap \Int$;\\
  $\quad$ $\quad$ $\quad$ $\quad$ $\testpoints(P,i,N) := \testpoints(P,i,N) \cap \{a\}$;\\
  $\quad$ $\quad$ $\quad$ if [$I \setminus \Int \neq \emptyset$] then \\
  $\quad$ $\quad$ $\quad$ $\quad$ pick $a \in I \setminus \Int$;\\
  $\quad$ $\quad$ $\quad$ $\quad$ $\testpoints(P,i,N) := \testpoints(P,i,N) \cap \{a\}$;\\
  $\quad$ $\quad$ for all $(R,k) \in \connectedArgs(P,i,N)$ with $|\aval(P,i,N)| = \infty$\\
  $\quad$ $\quad$ $\quad$ $\testpoints(R,k,N) := \testpoints(P,i,N)$;\\
  for all predicates $P$ and argument positions $i$ for $P$\\
  $\quad$ $\beta(P,i) := \testpoints(P,i,N) \cap \sort(P,i)^{\inta}$;
\end{tabular}
\end{center}

\subsubsection{Hammering a $\HBS(\SB)$ clause set into an $\HBS$ clause set}

\begin{center}
\begin{tabular}{l}
  $\text{HammerClauses}(N,\beta)$\\
  $\Pi_S := \emptyset$;\\
  $\tren_N(N) := \emptyset$; $\tfacts(N,\beta) := \emptyset$; $\sfacts(N,\beta) := \emptyset$;\\
  for all Horn clauses $\Lambda \parallel \Delta \rightarrow H \in N$\\
  $\quad$ $\Delta_T := \emptyset$; $\Delta_S := \emptyset$;\\ 
  $\quad$ for all variables $x$ in $\Lambda \parallel \Delta \rightarrow H$\\
  $\quad$ $\quad$ let $(Q_1,j_1), \ldots, (Q_1,j_m)$ be connected to $x$ in $\Delta \rightarrow H$;\\  
  $\quad$ $\quad$ $S(x) := \sort(x) \cap \bigcap_k \beta(Q_k,j_k)$;\\
  $\quad$ for all $\lambda \in \Lambda$\\
  $\quad$ $\quad$ let $x_1,\ldots,x_m$ be the variables in $\lambda$;\\
  $\quad$ $\quad$ let $P_{\lambda}$ be a fresh predicate symbol of arity $m$;\\
  $\quad$ $\quad$ $\Delta_T := \Delta_T \cup \{P_{\lambda}(x_1,\ldots,x_m)\}$;\\
  $\quad$ $\quad$ for all $(a_1,\ldots,a_m) \in S(x_1) \times \ldots \times S(x_m)$\\
  $\quad$ $\quad$ $\quad$ if[$\lambda \cdot \{x_1 \mapsto a_1, \ldots, x_m \mapsto a_m\}$ evaluates to \texttt{true}] \\
  $\quad$ $\quad$ $\quad$ $\quad$ $\tfacts(N,\beta) := \tfacts(N,\beta) \cup \{P_{\lambda}(a_1,\ldots,a_m)\}$;\\  
  $\quad$ if[$H = P(t_1,\ldots,t_n)$] then \\
  $\quad$ $\quad$ for all $t_i$ that are variables\\
  $\quad$ $\quad$ $\quad$ $S := \sort(P,i);$\\
  $\quad$ $\quad$ $\quad$ if [$x$ does not appear in $\Lambda$ and $\Delta$ or\\
  $\quad$ $\quad$ $\quad$ \hspace*{3ex} ($x$ has sort integer and only appears in argument positions\\
  $\quad$ $\quad$ $\quad$ \hspace*{5ex} $(Q,j)$ in $\Delta$ with sort real)] then \\
  $\quad$ $\quad$ $\quad$ $\quad$ $\Delta_S := \Delta_S \cup \{Q_{(P,i,S)}(t_i)\}$;\\
  $\quad$ $\quad$ $\quad$ $\quad$ if [$Q_{(P,i,S)} \not \in \Pi_S$] then\\
  $\quad$ $\quad$ $\quad$ $\quad$ $\quad$ $\Pi_S := \Pi_S \cup \{Q_{(P,i,S)} \}$;\\
  $\quad$ $\quad$ $\quad$ $\quad$ $\quad$ for all $a \in \beta(P,i)$ that belong to sort $S$\\
  $\quad$ $\quad$ $\quad$ $\quad$ $\quad$ $\quad$ $\sfacts(N,\beta) := \sfacts(N,\beta) \cup \{Q_{(P,i,S)}(a)\}$;\\
  $\quad$ $\tren_N(N) := \tren_N(N) \cup \{(\Delta_T,\Delta_S,\Delta \rightarrow H)\}$;\\
  return ($\tren_N(N)$; $\tfacts(N,\beta)$; $\sfacts(N,\beta)$);
\end{tabular}
\end{center}

\subsubsection{Hammering a universal conjecture into an $\HBS$ clause set}

\begin{center}
\begin{tabular}{l}
  $\text{HammerUnivConjecture}(\forall x_1,\ldots,x_n.Q(x_1,\ldots,x_n),\beta,\eta)$\\
  $\Delta_{\phi} := \emptyset$;\\
  for all variables $x_i$\\
  $\quad$ $S(x_i) := \beta(Q,i)$;\\
  $\quad$ if [$\sort(Q,i)^{\inta} \neq \bigcup_{a \in \beta(Q,i)} \eta(Q,i,a)$] then \\
  $\quad$ $\quad$ return $\emptyset$;\\
  for all $(a_1,\ldots,a_n) \in S(x_1) \times \ldots \times S(x_n)$\\
  $\quad$ $\Delta_{\phi} := \Delta_{\phi} \cup \{Q(a_1,\ldots,a_n)\}$;\\  
  return $\{\Delta_{\phi} \rightarrow \bot\}$;
\end{tabular}
\end{center}


\else
\vfill

{\small\medskip\noindent{\bf Open Access} This chapter is licensed under the terms of the Creative Commons\break Attribution 4.0 International License (\url{http://creativecommons.org/licenses/by/4.0/}), which permits use, sharing, adaptation, distribution and reproduction in any medium or format, as long as you give appropriate credit to the original author(s) and the source, provide a link to the Creative Commons license and indicate if changes were made.}

{\small \spaceskip .28em plus .1em minus .1em The images or other third party material in this chapter are included in the chapter's Creative Commons license, unless indicated otherwise in a credit line to the material.~If material is not included in the chapter's Creative Commons license and your intended\break use is not permitted by statutory regulation or exceeds the permitted use, you will need to obtain permission directly from the copyright holder.}

\medskip\noindent\includegraphics{cc_by_4-0.eps}

\fi

\end{document}